\documentclass[11pt]{article}

\usepackage{amssymb}
\usepackage[cmex10]{amsmath}
\usepackage{amsthm}
\usepackage{units, xcolor}
\usepackage{fullpage}
\usepackage{times}
\newif\ifconf
\conffalse

\DeclareMathOperator*{\argmax}{arg\,max}
\DeclareMathOperator*{\argmin}{arg\,min}
\DeclareMathOperator*{\esssup}{ess\,sup}

\DeclareMathOperator {\val}  {val}

\DeclareMathOperator {\Var}  {Var}

\newcommand {\set}   [1] {\left\{ #1 \right\}}
\newcommand {\brc}   [1] {\left(#1\right)}

\newcommand {\Exp}       {\mathbb{E}}
\newcommand {\Prob}  [1] {\Pr \brc{#1 }}

\newcommand {\E}     [1] {\Exp\left[#1\right]}

\newcommand {\Varr}  [1] {\Var \left[#1\right]}

\newcommand{\given}{\mid}

\newcommand {\bbZ}    {\mathbb{Z}}

\newcommand {\bbR}    {\mathbb{R}}

\newcommand {\calP}   {{\cal{P}}}
\newcommand {\calQ}  {{\cal{Q}}}
\newcommand {\calF}   {{\cal{F}}}

\newcommand {\calI}    {{\cal I}}
\newcommand {\calL}   {{\cal{L}}}

\newcommand {\Tau}    {{\cal{T}}}

\newcommand {\ord}    {{\cal{O}}} 

\newtheorem{theorem}{Theorem}
\newtheorem{lemma}[theorem]{Lemma}

\newtheorem{claim}[theorem]{Claim}
\newtheorem{corollary}[theorem]{Corollary}
\newtheorem{definition}[theorem]{Definition}

\newtheorem{remark}[theorem]{Remark}

\title{Satisfiability of Ordering CSPs Above Average\\ Is Fixed-Parameter Tractable}

\author{Konstantin Makarychev \\Microsoft Research \and Yury Makarychev
\\TTIC \and Yuan Zhou \\ MIT}
\date{}
\begin{document}
\maketitle

\begin{abstract}
We study the satisfiability of ordering constraint satisfaction problems (CSPs) above average. We prove the conjecture of Gutin, van Iersel, Mnich, and Yeo that the satisfiability above average of ordering CSPs of arity $k$ is fixed-parameter tractable for every $k$. Previously, this was only known for $k=2$ and $k=3$. We also generalize this result to more general classes of CSPs, including CSPs with predicates defined by linear inequalities.

To obtain our results, we prove a new Bonami-type inequality for the Efron---Stein decomposition. The inequality applies to functions defined on arbitrary product
probability  spaces. In contrast to other variants of the Bonami Inequality, it does not depend on the mass of the smallest atom in the probability space. We believe that this inequality is of independent interest.
\end{abstract}

\section{Introduction}

In this paper, we study the satisfiability of ordering constraint satisfaction problems (CSPs) above the average value. An ordering $k$-CSP
is defined by a set of variables $V = \set{x_1,\dots, x_n}$ and a set of constraints $\Pi$. Each constraint $\pi\in \Pi$
is a disjunction of clauses of the form $x_{\tau_1} < x_{\tau_2} < \dots < x_{\tau_k}$ for some distinct variables $x_{\tau_1},\dots, x_{\tau_k}$
from a $k$-element subset $V_{\pi}\subset V$. A linear ordering $\alpha$ of variables $x_1,\dots, x_n$ satisfies a constraint $\pi$ if one of the clauses in the disjunction agrees with the linear ordering $\alpha$. The goal is to find an assignment $\alpha$ that maximizes the number of satisfied constraints.

A classical example of an ordering CSP is the Maximum Acyclic Subgraph problem,  in which constraints are of the form ``$x_i < x_j$'' (the problem has arity 2). Another well-known example is the Betweenness problem, in which constraints are of the form ``$(x_i < x_j <x_k)$ or $(x_k < x_j <x_i)$'' (the problem has arity 3). Both problems are $NP$-hard and cannot be solved exactly in polynomial--time unless $P = NP$~\cite{Karp72, Opatrny79} (see also \cite{Seymour,CS,M-betw}).

There is a trivial approximation algorithm for ordering CSPs as well as other constraint satisfaction problems: output a random linear ordering of variables $x_1,\dots, x_n$ (chosen uniformly among all $n!$ linear orderings). Say, if each constraint is just a clause on $k$ variables, this algorithm satisfies each clause with probability $1/k!$ and thus satisfies a $1/k!$ fraction of all constraints in expectation.
In 1997, H\aa{}stad~\cite{Hastad} showed that for some regular (i.e., non-ordering) constraint satisfaction problems the best approximation algorithm is the \textit{random assignment} algorithm. His work raised
the following question: for which CSPs are there non-trivial approximation algorithms and
for which CSPs is the best approximation algorithm the random assignment?
This question has been extensively studied in the literature. Today, there are
many known classes of constraint satisfaction problems that do not admit non-trivial approximations
assuming the Unique Games or $P\neq NP$ conjectures (see e.g \cite{Hastad, AM, GR, Chan}).
There are also many constraint satisfaction problems for which we know non-trivial approximation algorithms.
Surprisingly, the situation is very different for ordering CSPs:
Guruswami, H\aa{}stad,
Manokaran, Raghavendra, and Charikar~\cite{GHMRC} showed that
\textit{all} ordering $k$-CSPs do not admit non-trivial approximation assuming the Unique Games Conjecture.

A similar question has been actively studied from the fixed-parameter tractability
perspective\footnote{We refer the reader to an excellent survey of results in this area by Gutin and Yeo~\cite{GY}.}~\cite{AGKSY, CFGJRTY, CGJRS,GKSY,KW,MR,RS}: Given an instance of a CSP, can we decide whether
$OPT\geq AVG+t$ for a fixed parameter $t$? Here, $OPT$ is the value of the optimal solution for the instance,
and $AVG$ is the expected value on a  random assignment.
In 2011, Alon, Gutin, Kim, Szeider, and Yeo~\cite{AGKSY} gave the affirmative answer to this question for all (regular, non-ordering) $k$-CSPs
with a constant size alphabet. In~\cite{GKSY, GIMY, GKMY}, Gutin et al. extended this result to
2-arity and 3-arity ordering CSPs, and conjectured~\cite{GIMY} that the satisfiabilty above average is fixed-parameter tractable for ordering CSPs of any arity $k$.
Below we state the problem formally.

\begin{definition}[Satisfiability of Ordering CSP Above Average]\label{def:above-avg}
Consider an instance $\cal I$ of arity $k$ and a parameter $t$. Let $OPT = OPT({\cal I})$ be the number of the constraints satisfied by the optimal solution, and
$AVG = AVG({\cal I})$ be the number of constraints satisfied in expectation by a random solution. We need to decide whether $OPT \geq AVG + t$.
\end{definition}
\begin{definition}
A problem with a parameter $t$ is fixed-parameter tractable	if there exists an
algorithm for the problem with running time $g(t) poly(n)$, where $g(t)$ is an arbitrary function of $t$, $poly$ is a fixed polynomial (independent of $t$),
and $n$ is the size of the input.
\end{definition}

In this paper, we prove the conjecture of Gutin et al.~\cite{GIMY} and show that the satisfiability above average of any ordering CSP of any arity $k$ is fixed-parameter tractable.
\begin{theorem}\label{thm:main-intro}
There exists a deterministic algorithm that given an instance $\calI$ of an ordering $k$-CSP on $n$ variables and a parameter $t$, decides whether
$OPT(\calI) \geq AVG(\calI) + t$ in time $g(t)poly_k(n)$, where $g$ is a function of $t$, $poly_k$ is a polynomial of $n$ ($g$ and $poly_k$ depend on $k$).
If $OPT(\calI) \geq AVG(\calI) + t$, then the algorithm also outputs an assignment satisfying at least $AVG(\calI) + t$ constraints.
\end{theorem}
Furthermore, we prove that the problem has a \emph{kernel} of size $O_k(t^2)$.

\medskip

\noindent \textbf{Techniques.} Let us examine approaches used previously for ordering CSPs.
The algorithms of Gutin et al.~\cite{GKSY, GIMY, GKMY} work by applying a carefully chosen
set of reduction rules to ordering CSPs of arity 2 and 3. These rules heavily depend on the structure of
$2$ and $3$ CSPs. Unfortunately, the structure of ordering
CSPs of higher arities is substantially more complex. Here is a quote from~\cite{GIMY}: \emph{``it appears technically
very difficult to extend results obtained for arities $r = 2$ and $3$ to $r > 3$}.''
In this paper, we do not use such reductions.

The papers~\cite{CMM, GZ, Mak} use an alternative approach to get an advantage over the random assignment
for special families of ordering CSPs. They first reduce the ordering $k$-CSP to a regular $k$-CSP with a constant size alphabet, and
then work with the obtained regular $k$-CSP. However, this reduction, generally, does not preserve the value of the CSP. So if for the original
ordering CSP instance $\calI$ we have $OPT({\cal I})\geq AVG({\cal I})+t$, then for the new instance $\calI'$ we may
have $OPT(\calI') \ll AVG(\calI)+t$ (we note that $AVG(\calI') = AVG(\cal I)$).
In this paper, we do not use this reduction either.

Instead, we treat all ordering CSPs as CSPs with the continuous domain: Our goal is to arrange
all variables on the interval $[-1,1]$ so as maximize the number of satisfied constraints.
The arrangement of variables uniquely  determines their order. Moreover, if we
independently assign random values from $[-1,1]$ to variables $x_i$, then the induced ordering on $x_i$'s will be
uniformly distributed among all $n!$ possible orderings.
Thus, our reduction preserves the values of $OPT$ and $AVG$. However, we can no longer
apply Fourier analytic tools used previously in~\cite{AGKSY, GZ, Mak}. We cannot use the
(standard) Fourier analysis on $[-1,1]^n$, since we have no control over the Fourier coefficients
of the functions we  need to analyze. Instead, we work with the Efron---Stein decomposition~\cite{ES}
(see Sections~\ref{sec:proof-overview} and \ref{sec:EfronStein}). We show that all terms in
the Efron---Stein decomposition have a special form. We use this fact to prove that an
ordering $k$-CSP that depends on many variables must have a large variance.
Specifically, we show that
if a $k$-CSP instance depends on $C_k t^2$ variables, then the standard deviation of its value from the mean
(on a random assignment) is greater than $c_k t$ (for some $C_k$ and $c_k\gg1$). As is, this claim does not imply that
$OPT \geq AVG + t$ since for some assignments the value may be substantially less than $AVG-t$.
To finish the proof of the main result, we prove a new hypercontractive inequality, which is
an analog of the Bonami Lemma~\cite{Bonami}. This inequality is one of the main
technical contributions of our paper.

\begin{theorem}[Bonami Lemma for Efron---Stein Decomposition]\label{thm:bonami-for-ES}
Consider $f \in L_2(\Omega^n, \mu^n)$. Let $f = \sum_{S} f_S$ be the Efron---Stein decomposition of $f$. Denote the degree of the decomposition by $d$.
Assume that for every set $S$,
\begin{equation}\label{eq:Bonami-Fourth-Moment}
\|f_S\|_4^4 \equiv \E{f_{S}^4} \leq C \|f_S\|_2^4 \equiv C \E{f_{S}^2}^2.
\end{equation}
Then
\begin{equation}\label{eq:Bonami-Bound}
\|f\|_4^4 \equiv \E{f(X_1,\dots, X_n)^4} \leq 81^d C \|f\|_2^4 \equiv 81^d C  \E{f(X_1,\dots, X_n)^2}^2.
\end{equation}
\end{theorem}
We note that hypercontractive inequalities have been extensively studied under various settings
(see e.g., \cite{Talagrand94, DS96, Wolf07, MOS}). However, all of them depend on
the \textit{mass of the smallest atom in the probability space}. In our case, the smallest atom is polynomially small in $n$,
which is why we cannot apply known hypercontractive inequalities.
This is also the reason why we need an extra condition~(\ref{eq:Bonami-Fourth-Moment}) on the function $f$.
Condition (\ref{eq:Bonami-Fourth-Moment}) is a ``local'' condition in
the sense that all expectations in~(\ref{eq:Bonami-Fourth-Moment}) are over at most $d$ variables for every set $S$.
Consequently, as we will see below, it is very easy to verify that it holds in many cases
(in contrast to (\ref{eq:Bonami-Bound}), which is very difficult to verify directly).
Note also that  condition (\ref{eq:Bonami-Fourth-Moment}) is necessary --- if it is not satisfied, then the ratio $\|f\|_4/\|f\|_2$ can
be arbitrarily large even for $d=1$.

\medskip

\noindent\textbf{Extensions.}
Once we assume that the domain of every variable is the interval $[-1,1]$, we might be tempted
to write more complex constraints than before such as
\emph{``the average of $x_1$, $x_2$ and $x_3$ is at most
$x_4$''}, or \emph{``$x_1$ lies to the left of the midpoint between $x_2$ and $x_3$''},
or \emph{``$x_1$ is closer to $x_2$ than to $x_3$''}.
 Each of these constraints can be written as a system of linear inequalities or a disjunction of clauses,
 each of which is given by a system of linear inequalities. For instance,
  \emph{``$x_1$ lies to the left of the midpoint between $x_2$ and $x_3$''}
  can be written as  $2x_1 - x_2 -  x_3 < 0$.
In Section~\ref{sec:PPP}, we extend our results to CSPs in which every constraint is a disjunction of clauses, each of which is
a ``small'' linear program (LP).
Namely, each constraint should have arity at most $k$, only variables that a constraint depends on should appear in the LPs that define it,
and all LP coefficients must be integers in the range $\{-b,\dots,b\}$ (for a fixed $b$).
We call this new class of CSPs $(k,b)$-LP CSPs.

\begin{definition}\label{def:LP-CSP}
A $(k,b)$-LP CSP is defined by a set of variables $V = \set{x_1,\dots, x_n}$ taking values in the interval $[-1,1]$
and a set of constraints $\Pi$. Each constraint $\pi\in \Pi$ is a disjunction of clauses of the form
$Ax < c$, where $A$ is a matrix with integer coefficients in the range $[-b,b]$;
$c$ is a vector with integer coefficients in the range $[-b,b]$; the indices of  non-zero columns of the matrix
$A$ lie in the set $V_{\pi}$ of size $k$ (the set $V_{\pi}$ is the same for all clauses in $\pi$).
The goal is to assign distinct real values to variables $x_i$ so as to maximize the number of satisfied constraints.
\end{definition}

In fact, we extend our results to a much more general class of valued CSPs -- all CSPs whose
predicates lie in a lattice of functions with some natural properties
(see Sections~\ref{sec:general-framework} and~\ref{sec:PPP}  for details);
but we believe that the subclass of $(k,b)$-LP CSPs is the most natural example of CSPs in the class.
Observe that every ordering $k$-CSP is a $(k,1)$-LP CSP since we can
write each clause $x_1< x_2< \dots < x_k$ as the system of linear equations
$x_i - x_{i+1} < 0$ for $i\in \{1,\dots,k-1\}$. Similarly, every $k$-CSP on
a finite domain $\{1,\dots, d\}$ is equivalent to a $(k,d)$-LP CSP. The reduction works as follows:
We break the interval $[-1,1]$ into $d$ equal subintervals $((2j-d - 2)/d, (2j-d)/d )$
and map every value $j$ to the $j$-th interval. Then, we replace every
condition $x_i = j$ with the equation $x_i \in ((2j-d-2)/d, (2j-d)/d)$
which can be written as $-d x_i < -(2j-d-2)$ and $d x_i < (2j-d)$.

\medskip

\noindent \textbf{Overview.} In the next section we give an informal overview of the proof.
We formally define the problem and describe the Efron---Stein decomposition in Section~\ref{sec:prelim}.
Then, in Section~\ref{sec:ES-ordering-CSP}, we prove several claims about the Efron---Stein decomposition
of ordering CSPs. We derive the main results (Theorem~\ref{thm:main-intro} and Theorem~\ref{thm:main}) in
Section~\ref{sec:proof-main-thm}. Finally, we prove the Bonami Lemma for the Efron---Stein decomposition
in Section~\ref{sec:Bonami}. We generalize our results to all CSPs with a lattice structure in Section~\ref{sec:general-framework}
and show that $(k,b)$-LP CSPs (as well as more general ``piecewise polynomial'' CSPs)
have a lattice structure in Section~\ref{sec:PPP}.

\section{Proof Overview}\label{sec:proof-overview}
Our high-level approach is similar to that developed by Alon et al.~\cite{AGKSY} and Gutin et al.~\cite{GKSY, GKMY, GIMY}.
As in~\cite{GKSY, GKMY, GIMY}, we design an algorithm that given an instance $\cal I$ of an ordering CSP does the following:
\begin{itemize}
\item It either finds a kernel (another instance of the ordering CSP) $\cal K$ on $O(t^2)$ variables such that $OPT({\cal I}) = OPT({\cal K})$ and $AVG({\cal I}) = AVG({\cal K})$.
Then we can decide whether $OPT \geq AVG + t$ by trying out all possible permutations of variables that $\cal K$ depends on in time $\exp(O(t^2 \log t))$.
\item Or it certifies that $OPT({\cal I}) \geq AVG({\cal I}) + t$.
\end{itemize}
To this end, we show that either $\cal I$ depends on at most $O(t^2)$ variables or the variance of $\val_{\cal I} (\alpha)$ is $\Omega(t^2)$ (where $\alpha$ is chosen uniformly at random). In the former case, the restriction of $\cal I$ to the variables it depends on is the desired kernel of size $O(t^2)$. In the latter case, $OPT \geq AVG + t$.
Though our approach resembles that of~\cite{GKSY, GKMY, GIMY} at the high level, we employ very different techniques to prove our results.

We extensively use Fourier analysis and, specifically, the Efron---Stein decomposition. Fourier analysis is a very powerful tool, which works especially well with product spaces. The space of feasible solutions of an ordering CSP is not, however, a product space --- it is a discrete space
that consists of $n!$ linear orderings of variables $x_1,\dots, x_n$. To overcome this problem, we define ``continuous solutions'' for an ordering
CSP (see Section~\ref{sec:csp-over-product-space}). A solution  is an assignment of real values in $[-1,1]$ to variables $x_1,\dots, x_n$; that is, it is a point in the product space $[-1,1]^n$. Each continuous solution defines a combinatorial solution $\alpha$ in a natural way: $\alpha$ orders variables $x_1,\dots, x_n$ according to the values assigned to them (e.g., if we assign values $-0.5$, $-0.9$ and $0.5$ to $x_1$, $x_2$ and $x_3$ then $x_2 < x_1 < x_3$ according to $\alpha$). Thus we get an optimization problem over the product space $[-1,1]^n$.
Denote by $\Phi:[-1,1]^n \to {\mathbb R}$ its objective function.
We consider the Efron---Stein decomposition of $\Phi$: $\Phi = \sum_{S: |S|\leq k} \Phi_S$ (see Section~\ref{sec:EfronStein}). Here, informally, $\Phi_S$ is the part of $\Phi$ that depends on variables $x_i$ with $i\in S$. All functions $\Phi_S$ are uncorrelated: $\E{\Phi_{S_1} \Phi_{S_2}} = 0$ for $S_1 \neq S_2$. We show that each $\Phi_S$ is either identically equal to $0$ or has variance greater than some positive number, which depends only on $k$ (see Section~\ref{sec:ES-ordering-CSP}, Lemma~\ref{lem:beta}).
We now consider two cases.

I. In the first case, there are at most $\Theta_k(t^2)$  terms $\Phi_S$  not equal to $0$. Using that $\Phi_S$ depends only on variables $x_i$ with $i\in \bigcup_{\Phi_S \neq 0} S$ and that there are at most $\Theta_k(t^2)$ sets $S$ such that $\Phi_S \neq 0$, we get that $\Phi$ depends on at most $O_k(t^2)$ variables and we are done.

II. In the second case, there are at least $c_k t^2$  terms $\Phi_S$  not equal to $0$. Since the variance of each term $\Phi_S$ is $\Theta_k(1)$ and all terms $\Phi_S$ are uncorrelated, the variance of $\Phi$ is at least $\Theta_k(t^2)$ (see Theorem~\ref{thm:variance}). Therefore, $\Phi$ deviates from $AVG = \E{\Phi}$ by at least $\sqrt{\Var{\Phi}} = \Theta(t)$. We then show that $\Phi-\Exp\Phi$ satisfies the conditions of Theorem~\ref{thm:bonami-for-ES} (see Lemma~\ref{lem:bound-on-C}) and the degree of the decomposition is at most $k$.
Thus, the ratio $\|\Phi-\Exp\Phi\|_4/\|\Phi-\Exp\Phi\|_2$ is bounded by $O_k(1)$. This implies by a theorem of Alon, Gutin, Kim, Szeider, and Yeo~\cite{AGKSY} (see Theorem~\ref{thm:AGKSY} in Section~\ref{sec:proof-main-thm})
that $\Pr (\Phi-\Exp\Phi > t)> 0$. Hence, $OPT \geq AVG + t$.




\section{Preliminaries}\label{sec:prelim}
\subsection{Ordering CSP}\label{sec:prelim-order-csp}
Consider a set of variables $V = \set{x_1,\dots, x_n}$.
An ordering constraint $\pi$ on
a subset of variables $x_{i_1}, \dots, x_{i_k}$
is a set of linear orderings of
$x_{i_1}, \dots, x_{i_k}$.
A linear ordering $\alpha$ of $V$  satisfies a constraint $\pi$ on $x_{i_1}, \dots, x_{i_k}$ if the restriction of $\alpha$ to
$x_{i_1}, \dots, x_{i_k}$ is in $\pi$.
We say that $\pi$ depends on variables $x_{i_1}, \dots, x_{i_k}$.
\begin{definition}
An instance $\cal I$ of an ordering constraint satisfaction problem consists of a set of variables $V = \set{x_1,\dots, x_n}$ and a set of constraints $\Pi$;
each constraint $\pi \in P$ depends on some subset of variables.
A feasible solution to $\cal I$ is a linear ordering of
variables $x_1,\dots, x_n$. The value $\val(\alpha) = \val_{\cal I} (\alpha)$ of a solution $\alpha$ is the number of constraints in $\Pi$ that $\alpha$ satisfies. The goal of the problem is to find a solution of maximum possible value.
\end{definition}

We denote the value of the optimal solution by $OPT$:
$$OPT = \max_{\alpha} \val(\alpha).$$
The average value $AVG$ of an instance is the expected value of a solution chosen uniformly at random among $n!$ feasible solutions:
$$AVG = \Exp_{\alpha}[\val(\alpha)].$$
We say that $\cal I$ has arity $k$ if each constraint
in $\cal I$ depends on at most $k$ variables.

\begin{definition}
In the Satisfiability Above Average Problem,
we are given an instance of an ordering CSP of arity $k$
and a parameter $t$. We need to decide if there is a solution $\alpha$ that satisfies at least $AVG + t$ constraints, or, in other words, if $OPT \geq  AVG + t$.
\end{definition}

In this paper, we show that this problem is fixed-parameter tractable. To this end, we design an algorithm
that either finds a kernel on $O(t^2)$ variables or certifies that $OPT \geq  AVG + t$.
\begin{theorem}\label{thm:main}
There is an algorithm that given an instance of an
ordering CSP problem of arity $k$ and a parameter $t$,
either finds a kernel on at most $\kappa_k t^2$ variables
(where constant $\kappa_k$ depends only on $k$)
or
certifies that $OPT \geq  AVG + t$.
The algorithm runs in time $O_k(m+n)$ linear in the number of constraints $m$ and variables $n$ (the coefficient in the $O$-notation depends on $k$).	
\end{theorem}

\subsection{Ordering CSPs over $[-1,1]^n$}\label{sec:csp-over-product-space}

Consider an instance $\cal I$ of an ordering CSP on variables $x_1,\dots, x_n$.
Let us say that a continuous feasible solution to $\cal I$ is an assignment of distinct values $\hat x_1, \dots, \hat x_n \in [-1,1]$ to variables $x_1,\dots, x_n$. Each continuous solution $\hat x_1, \dots, \hat x_n$ defines
an ordering $\alpha$ of variables $x_i$: $x_a$ is less then $x_b$ with respect to $\alpha$ if and only if $\hat x_a < \hat x_b$.
We define the value of a continuous solution $\hat x_1, \dots, \hat x_n$
as the value of the corresponding solution (linear ordering) $\alpha$. We will denote the value of
solution $\hat x_1, \dots, \hat x_n$ by
$\Phi(\hat x_1, \dots, \hat x_n)$.

Note that if we sample a continuous solution
$\hat x_1, \dots, \hat x_n$ uniformly at random, by  choosing values $\hat x_i$ independently and uniformly from $[-1,1]$, the corresponding solution $\alpha$
will be uniformly distributed among $n!$ feasible solutions. Therefore,
\begin{equation*}
OPT = \max_{\hat x_1, \dots, \hat x_n\in [-1,1]} \Phi(\hat x_1, \dots, \hat x_n)\quad\text{and}\quad
AVG = \Exp_{\hat x_1, \dots, \hat x_n\in [-1,1]}
\Phi(\hat x_1, \dots, \hat x_n). 	
\end{equation*}
Note that all $\hat x_i$ are distinct a.s. 
and thus a random point in $[-1,1]^n$ is a feasible
continuous solution a.s. 

\subsection{Efron---Stein Decomposition}\label{sec:EfronStein}
The main technical tool in this paper is the Efron---Stein decomposition. We refer the reader to~\cite[Section 8.3]{ODonnell} for
a detailed description of the decomposition. Now, we just remind its definition and basic properties.

The Efron---Stein decomposition can be seen as a generalization of the Fourier expansion
of Boolean functions on the Hamming cube $\{\pm 1\}^n$. Consider the Fourier expansion of a function
$f:\{\pm 1\} \to \mathbb R$,
$$f(x_1,\dots,x_n) = \sum_{S\subset \{1,\dots, n\}} \hat f_S \chi_S(x_1,\dots, x_n),$$
where $\hat f_S$ are Fourier coefficients of $f$. Informally,
the Fourier expansion  breaks $f$ into pieces, $\hat f_S \chi_S(x_1,\dots, x_n)$, each of which depends on its own set of variables: The term
$\hat f_S \chi_S(x_1,\dots, x_n)$ depends on variables $\{x_i:i \in S\}$ and no other variables.

The Efron---Stein decomposition is an analogue of the Fourier expansion for functions defined on arbitrary product probability spaces. Consider a probability space $(\Omega,\mu)$ and the product probability space $(\Omega^n,\mu^n)$.
Let $f:\Omega^n\to \mathbb R$ be a function (random variable) on $\Omega^n$. Informally, the Efron---Stein decomposition of $f$ is the decomposition of $f$ into the sum of
functions $f_S$,
$$f = \sum_{S\subset \{1,\dots, n\}} f_S,$$
in which $f_S:\Omega^n\to \mathbb R$ depends on variables $\set{x_i :i\in S}$.

We formally define the  Efron---Stein decomposition  as follows.
Consider the space $L_2(\Omega^n,\mu^n)$ of functions on $\Omega^n$ with bounded second moment. Note that $L_2(\Omega^n,\mu^n) = \bigotimes_{i=1}^n L_2(\Omega,\mu)$.
That is, every $f\in L_2(\Omega^n,\mu^n)$ can be represented as
$$f(x_1,\dots, x_n) = \sum_j f^j_1(x_1)\cdot f^j_2(x_2)\cdots  f^{j}_n(x_n),$$
for some functions $f^j_i\in L_2(\Omega,\mu)$. Let $\Lambda_0\subset L_2(\Omega,\mu)$ be the one-dimensional space of constant functions on $\Omega$. Let $\Lambda_{\perp} \subset  L_2(\Omega,\mu)$ be the orthogonal complement to $\Lambda_0$. That is, $\Lambda_{\perp}$ is the space of functions
$f\in L_2(\Omega,\mu)$ with $\E{f} = 0$. We have, $L_2(\Omega,\mu) = \Lambda_0 \oplus \Lambda_{\perp}$ and
$$L_2(\Omega^n,\mu^n) = \bigotimes_{i=1}^n L_2(\Omega,\mu) = \bigotimes_{i=1}^n (\Lambda_0 \oplus \Lambda_{\perp}).$$
Expanding this decomposition, we get a representation of $L_2(\Omega^n,\mu^n)$ as the direct sum of $2^n$ spaces:
$$L_2(\Omega^n,\mu^n)= \bigoplus_{S\subset \{1,\dots, n\}} V_S,$$
where $V_S$ is the closed linear span of the set of functions of the form $\prod_{i=1}^n f_i(x_i)$ where $f_i\in \Lambda_{\perp}$ if $i\in S$, and $f_i\in \Lambda_0$ if $i\notin S$.
Since functions in $\Lambda_0$ are constants, $V_S$ equals the closed linear span of the set of functions of the form $\prod_{i\in S} f_i(x_i)$ where
$f_i\in \Lambda_{\perp}$.

Consider a function $f\in L_2(\Omega^n,\mu^n)$. Let $f_S$ be the orthogonal projection of $f$ onto $V_S$.
Since the linear spaces $V_S$ are orthogonal, we have
$$f = \sum_{S \subset \{1,\dots, n\}} f_S.$$
We call this decomposition the Efron---Stein decomposition of $f$. We define the degree of $f$ as $\max \{|S|: f_S \neq 0\}$,
the size of the largest subset $S$ such that  $f_S$ is not identically equal to $0$ (we let the degree of $0$ to be $0$).

Let $(X_1,\dots,X_n)$ be a random element of $\Omega^n$. That is, $X_1,\dots,X_n$ are $n$ independent random elements of $\Omega$; each of them is distributed according to $\mu$. We write $f = f(X_1,\dots, X_n)$.
We will employ the following properties of the Efron---Stein decomposition (see \cite[Section 8.3]{ODonnell}).
\begin{enumerate}
\item $f_S(x_1,\dots, x_n)$ depends only on variables $x_i$ with $i \in S$.
\item For every two sets $S$ and $T$, $S\neq T$, we have $\E{f_S f_T} = 0$.
\item Let $S_1, \dots, S_r$ be subsets of $\{1,\dots, n\}$. Suppose that there is an index $j$ that belongs to exactly
one set $S_i$. Then $\E{\prod_{i=1}^r f_{S_i} } = 0$.
\end{enumerate}

We will also use the following equivalent and more explicit definition of the Efron---Stein decomposition.
For every subset $S$ of indices $\{1,\dots, n\}$, let
\begin{align}
f_{\subset S} &= \E{f(X_1,\dots, X_n)| \text{ all } X_i \text{ with } i \in S},\label{eq:ES-1}\\
f_S &= \sum_{T\subset S} (-1)^{|S\setminus T|} f_{\subset T}.\label{eq:ES-2}
\end{align}

\section{Efron---Stein Decomposition of Ordering CSP Objective}\label{sec:ES-ordering-CSP}
In this section, we study the Efron---Stein
decomposition of the function $\Phi(x_1,\dots, x_n)$.
To this end, we represent $\Phi(x_1,\dots, x_n)$ as a sum of
``basic ordering predicates'' and then analyze the Efron---Stein of a basic ordering predicate.

\subsection{Basic Ordering Predicate}\label{subsec:bop}
Let $\tau = (\tau_1,\dots,\tau_r)$ be a tuple of distinct
indices from $1$ to $n$.
Define the basic ordering predicate $\phi_{\tau}$ for $\tau$,
$$\phi_\tau(x_1,\dots, x_n) =
\begin{cases}
1, & \text{ if } x_{\tau_1} < x_{\tau_2} < \dots < x_{\tau_r},\\
0, & \text{ otherwise.}
\end{cases}
$$
Note that the indicator of each constraint
$\pi$ is a sum of ordering predicates:
$$\sum_{\tau: \text{ ordering } x_{\tau_1} < x_{\tau_2} < \dots < x_{\tau_r} \text{ is in } \pi} \phi_\tau(x_1,\dots, x_n),$$
where the sum is over permutations of variables that
the constraint $\pi$ depends on.
Since $\Phi$ is the sum of indicators of all
predicates $\pi$ in $\Pi$, $\Phi$ is also a sum
of basic ordering predicates $\phi_{\tau}$:
$$\Phi(x_1,\dots,x_n) = \sum_{\tau \in \Tau}
\phi_\tau(x_1,\dots, x_n) \ifconf.\else,\fi$$
 for some multiset $\Tau$.

\subsection{Efron---Stein Decomposition of Ordering Predicates}
Let $\Omega = [-1,1]$ and $\mu$ be the uniform measure on $\Omega$. We study the Efron---Stein decomposition
of a basic ordering predicate $\phi_{\tau}$.

\ifconf\pagebreak\fi
\begin{theorem}\label{thm:ES-ordering-predicate}
Let $\tau$ be a tuple of distinct indices of size $d \leq k$. Denote $g=\phi_{\tau}$. Consider the Efron---Stein decomposition of $g$, $g = \sum g_S$, over $[-1,1]^n$ with uniform measure.
There exists a set of polynomials $q_{S,\tau'}$ with integer coefficients of degree at most $d$ such that
$$g_S(x_1,\dots, x_n) = \sum_{\tau'}
\phi_{\tau'}(x_1,\dots, x_n) \frac{q_{S,\tau'}(x_1,\dots, x_n)}{2^{d} d!},$$
where the summation is over all permutations $\tau'$ of $S$. The polynomial $q_{S,\tau'}$ depends only on variables in $\set{x_i:i\in S}$. It is equal to $0$ if
$S$ is not a subset of $\set{\tau_1,\dots,\tau_d}$.
\end{theorem}
\begin{proof}
We may assume without loss of generality that
$\tau=\{1,2,\dots, d\}$.
Since $g$ depends only on
variable $x_1,\dots, x_d$, $g_S \neq 0$ only if $S \subset \set{1,\dots, d}$. We may therefore assume that $n = d$
for notational convenience.

Denote the elements of $S$ by $s_1 < s_2 < \dots < s_t$.
Define auxiliary variables
$X_0 = -1$ and $X_{d+1} = 1$, and let $s_0 = 0$ and $s_{t+1} = d+1$.
Let ${\cal O}_{ab}$ be the indicator of the event that
$X_i < X_j$ for every $a\leq i < j \leq b$.
Then $g = \ord_{1,d}$.  Note that
$g = \ord_{1,d} = \prod_{i=0}^{t} \ord_{s_i,s_{i+1}}$.
All events for $\ord_{s_i,s_{i+1}}$ (for $i\in \{0,\dots, t\}$)
are independent given variables $X_{s_1},\dots, X_{s_t}$. Therefore,
\begin{equation}\label{eq:g-subset-S}
g_{\subset S} = \E{g \given X_{s_1},\dots, X_{s_t}} =
\prod_{i=0}^{t} \E{\ord_{s_i,s_{i+1}} \given X_{s_1},\dots, X_{s_t}}.
\end{equation}
For each $i$, we have
$$\E{\ord_{s_i,s_{i+1}} \given X_{s_1},\dots, X_{s_t}} = \E{\ord_{s_i,s_{i+1}} \given X_{s_i}, X_{s_{i+1}}}= \Prob{\ord_{s_i,s_{i+1}} = 1\given X_{s_i}, X_{s_{i+1}}}.$$
If $X_{s_i} \geq X_{s_{i+1}}$, then $\Prob{\ord_{s_i,s_{i+1}} = 1 \given X_{s_i}, X_{s_{i+1}}}= 0$. Otherwise,
$$\Prob{\ord_{s_i,s_{i+1}} = 1 \given X_{s_i}, X_{s_{i+1}}}=
\left(\frac{X_{s_{i+1}} - X_{s_{i}}}{2}\right)^{s_{i+1} - s_{i}-1}\cdot \frac{1}{(s_{i+1} - s_{i}-1)!}.$$
We computed the probability above as follows: Given $X_{s_i}\leq X_{s_{i+1}}$, the probability that
$X_j\in [X_{s_i}, X_{s_{i+1}}]$ for all $j\in \{s_i,\dots, s_{i+1}\}$
equals $\big((X_{s_{i+1}} - X_{s_{i}})/2\big)^{s_{i+1}-s_i -1}$. Then, given that $X_{s_i}\leq X_{s_{i+1}}$
and $X_j\in [X_{s_i}, X_{s_{i+1}}]$ for all $j\in \{s_i,\dots, s_{i+1}\}$, the probability that
$X_{s_i+1}\leq \dots \leq X_{s_{i+1}-1}$ equals $1/(s_{i+1}-s_i -1)!$ as all orderings
of $X_{s_i+1}, \dots, X_{s_{i+1}-1}$ are equally likely.
We get
$$\E{\ord_{s_i,s_{i+1}} \given X_{s_1},\dots, X_{s_t}} = I\set{X_{s_i} < X_{s_{i+1}}} \frac{\left(X_{s_{i+1}} - X_{s_{i}}\right)^{s_{i+1} - s_{i}-1}}{2^{s_{i+1} - s_{i}-1} (s_{i+1} - s_{i}-1)!}.$$
Plugging this expression in (\ref{eq:g-subset-S}), we obtain the following formula
\begin{align*}
	g_{\subset S} &= \prod_{i=0}^t I\set{X_{s_i} < X_{s_{i+1}}} \frac{\left(X_{s_{i+1}} - X_{s_{i}}\right)^{s_{i+1} - s_{i}-1}}{2^{s_{i+1} - s_{i}-1} (s_{i+1} - s_{i}-1)!} \\
	&=
I\set{X_{s_1}<X_{s_2}< \dots X_{s_t}}\,
\frac{\prod_{i=0}^t \left(X_{s_{i+1}} - X_{s_{i}}\right)^{s_{i+1} - s_{i}-1}}{2^{d-|S|}
\prod_{i=0}^t (s_{i+1} - s_{i}-1)!}.
\end{align*}
Observe that $\prod_{i=0}^t (s_{i+1} - s_{i}-1)$ divides
$(\sum_{i=0}^t (s_{i+1} - s_{i}-1))! = (d - |S|)!$.
Thus
$$\frac{d!}{\prod_{i=0}^t (s_{i+1} - s_{i}-1)!}\,2^{|S|}\,\prod_{i=0}^t \left(X_{s_{i+1}} - X_{s_{i}}\right)^{s_{i+1} - s_{i}-1}$$
is a polynomial with integer coefficients of degree at most  $d -|S|$. Denote this polynomial by $p_S$.
Then,
$$
g_{\subset S} = I\set{X_{s_1}<X_{s_2}< \dots X_{s_t}}
\frac{p_{S}(X_1,\dots, X_d)}{2^{d} d!}\\
=
\sum_{\tau'}
\phi_{\tau'}(X_1,\dots, X_d) \frac{p_{S}(X_1,\dots, X_d)}{2^{d} d!},
$$
where the sum is over all permutations $\tau'$ of $\{1,\dots, d\}$.
Using the identity
$
f_S = \sum_{T\subset S} (-1)^{|S\setminus T|} f_{\subset T}$,
we get a representation of $S$ as
$$f_S = \sum_{\tau'}
\phi_{\tau'}(X_1,\dots, X_d) \frac{q_{S,\tau}(X_1,\dots, X_d)}{2^{d} d!},$$
where $q_{S,\tau'}$ are some polynomials with integer coefficients.
\end{proof}
Since $\Phi$ is a sum of some basic ordering predicates (see Section~\ref{subsec:bop}), we get the following corollary.

\begin{corollary}\label{cor:ES-objective}
Let $\cal I$ be an instance of an ordering CSP problem of arity
at most $k$. Let $\Phi(x_1,\dots,x_n)$ be the value of continuous solution $(x_1,\dots,x_n)$.
Then the Efron---Stein decomposition of $\Phi$ has degree at most $k$. Moreover there exist
polynomials $q_{S,\tau}$ with integer coefficients of degree at most $k$ such that
$$\Phi_S(x_1,\dots,x_n) = \sum_{\tau\in \Tau'}
\phi_{\tau}(x_1,\dots,x_n) \frac{q_{S,\tau}(x_1,\dots, x_n)}{2^{k} k!},$$
where the summation is over some set $\Tau'$ of tuples of indices in $S$, and $q_{S,\tau}$ depends only on \ifconf\else variables in\fi $\set{x_i:i \in S}$.
\end{corollary}
\subsection{Variance of Ordering CSP Objective}\label{sec:variance-CSP}
In this section, we show that the variance $\Varr{\Phi} = \Omega(\nu)$ if $\Phi$  (non-trivially) depends on at least $\nu$ variables.

\begin{claim}\label{claim:compact}
There exists a sequence of positive numbers $\alpha_d$ such that for every polynomial $f(x_1,\dots,x_d)$ of degree at most $d$ with integer coefficients we have $\E{\phi_{1,\dots,d} f^2} \geq \alpha_d$.
\end{claim}
\begin{proof}
Consider the set $\calQ$ of polynomials over $x_1, \dots, x_d$
of degree at most $d$. Let $\calQ_1$ be the set of polynomials in $\calQ$, whose largest
in absolute value coefficient is equal to $1$ or $-1$.

Denote $V(f) = \E{\phi_{1,\dots,d} f^2}$. For every $f \in \calQ_1$, we have $V(f) > 0$ since $f$ is not identically equal to $0$ on $\{x_1\leq x_2\leq \dots \leq x_d\}$.
Note that $\calQ_1$ is a compact set and $V(f)$ is a continuous function on it.
Therefore, $V$ attains its minimum on $Q_1$.
Let $\alpha_d = \min_{f \in \calQ_1} V(f) > 0$.

Now let $f$ be a polynomial with integer coefficients of degree at most $d$. Denote the absolute value of its largest
coefficient (in absolute value) by $M$. $M$ is a positive integer and thus $M \geq 1$. We have $f/M \in Q_1$ and thus
\ifconf
$V(f) = M^2 \cdot V(f/M) \geq M^2 \alpha_d \geq \alpha_d$.
\else
$$V(f) = M^2 \cdot V(f/M) \geq M^2 \alpha_d \geq \alpha_d.$$
\fi
\end{proof}
\begin{lemma}\label{lem:beta}
The following claim holds for some positive parameters $\beta_k$. Let $\cal I$ be an instance of arity at most $k$. Let $\Phi = \sum_{S} \Phi_S$ be the Efron---Stein decomposition of $\Phi$. Then for every set $S$ either $\Phi_S = 0$ or $\E{\Phi_S^2} \geq \beta_k$.
\end{lemma}
\begin{proof}
Let $\beta_d = \alpha_d / (2^{k} k!)^2 > 0$, where $\alpha_d$ is as in Claim~\ref{claim:compact}.
Assume that $\Phi_S \neq 0$. By Corollary~\ref{cor:ES-objective},
$$\Phi_S(x_1,\dots,x_d) = \sum_{\tau\in \Tau'}
\phi_{\tau}(x_1,\dots,x_d) \frac{q_{S,\tau}(x_1,\dots, x_n)}{2^{k} k!}.$$
Note that all functions
$\phi_{\tau}(x_1,\dots,x_d)   q_{S,\tau}(x_1,\dots, x_n)/ (2^{k} k!)$ have disjoint support, and, therefore, are pairwise orthogonal.
Choose one tuple $\tau \in \Tau'$ such that $q_{S,\tau}\neq 0$.
We have, $$\E{\Phi_S^2} \geq \E{\phi_{\tau} q_{S,\tau}^2 / (2^{k} k!)^2} = \E{\phi_{\tau} q_{S,\tau}^2} / (2^{k} k!)^2.$$
By Claim~\ref{claim:compact}, $\E{\phi_{\tau} q_{S,\tau}^2} \geq \alpha_d$ and hence
$\E{\Phi_S^2} \geq \alpha_d / (2^{k} k!)^2$.
\end{proof}


We say that $\Phi$ depends on the variable $x_i$ if there exist two vectors $x$ and $x'$ that differ only in the $i$-th
coordinate such that $\Phi(x) \neq \Phi(x')$.

\begin{theorem}\label{thm:variance}
Let $\cal I$ be an instance of arity at most $k$. Suppose that $\Phi$ depends on
at least $\nu$ variables. Then  $\Varr{\Phi} \geq \nu \beta_k/k$.
\end{theorem}
\begin{proof}
Consider the Efron---Stein decomposition of $\Phi$. Let
$V' = \bigcup_{S: \Phi_S \neq 0} \set{x_i:i\in S}$.
Note that $\Phi$ depends on all variables in $V'$ and no other variables. Thus, $|V'|\geq\nu$.
There are at least $\nu / k$ non-empty sets $S$ with $\Phi_S \neq 0$ since each such set $S$ contributes at most $k$ elements to $V'$.
For $S\neq \varnothing$, we have $\E{\Phi_S} = 0$  and hence $\Varr{\Phi_S} = \E{\Phi_S^2}$. By Lemma~\ref{lem:beta},
$\Varr{\Phi_S} = \E{\Phi_S^2} \geq \beta_k$ if $\Phi_S \neq 0$ and $S\neq \varnothing$. We have,
\begin{equation*}\Varr{\Phi} = \sum_S \Varr{\Phi_S}
\geq |\set{S\neq \varnothing:\Phi_S \neq 0}| \beta_k
\geq (\nu / k) \beta_k\ifconf.\tag*{\qedhere}\fi
\end{equation*}
\ifconf\else
as required.
\fi
\end{proof}

\section{Bonami Lemma for ordering CSPs}
We are going to apply Theorem~\ref{thm:bonami-for-ES} (the Bonami Lemma for the Efron---Stein decomposition) to the function $f = \Phi - \E{\Phi}$, where $\Phi$ is the objective function of the ordering CSP problem. In order to do that, we show now that $f$ satisfies the condition of the theorem
(Condition~(\ref{eq:Bonami-Fourth-Moment})) with some constant $C$ that depends only on the arity of the CSP.

\begin{lemma}\label{lem:bound-on-C}
There exists a sequence of constants $C_k$ such that the following holds. Let $\cal I$ be an instance of an ordering CSP of arity at most $k$. Let $f = \Phi - \E{\Phi}$ and $S \subset \set{1,\dots, k}$ of cardinality at most $k$. Then	
$$\E{f_{S}^4} \leq C_k  \E{f_{S}^2}^2.$$
\end{lemma}
\begin{proof}
We assume that $S$ is non-empty as otherwise both the left and right hand sides of the inequality are equal to $0$
(since $f_{\varnothing} = \E{f} = 0$). Therefore, $f_{S} = \Phi_{S}$.

Since $|S| \leq k$, we may assume without loss of generality that $S \subset \set{1,\dots, k}$.
Let $\calQ$ be the set of all functions
of $x_1,\dots, x_{k}$ of the form
$$\sum \phi_{\tau}(x_1,\dots,x_k) q_{S,\tau}(x_1,\dots, x_k),$$
where $q_{S,\tau}$ are some polynomials of degree at most $k$ with real coefficients. By Corollary~\ref{cor:ES-objective}, $f_{S} \in \calQ$. Let $\calQ_1 =\set{h \in \calQ: \|h\|_2 = 1}$. Note that $\calQ_1$ is
a compact set (since $\calQ$ is a finite dimensional space; and $\|\cdot\|_2$ is a norm on it). Therefore, the continuous
function
$W(g) = \E{g^4}$
is bounded when $g \in\calQ_1$. Denote its maximum by $C_k$
(note that $C_k$ depends only on $k$ and not on $\cal I$).

Letting $g=  f_{S} / \|f_{S}\|_2 \in \calQ_1$,  we have,
$$\E{f_S^4} = \|f_{S}\|_2^4 \cdot W(g) \leq C_k \cdot \|f_{S}\|_2^4.$$
as required.
\end{proof}


\section{Proof of Main Theorems}\label{sec:proof-main-thm}
In this section, we prove Theorems~\ref{thm:main} and~\ref{thm:main-intro}. We will need the following theorem.
\begin{theorem}[Corollary 1 from Alon, Gutin, Kim, Szeider, and Yeo~\cite{AGKSY}]\label{thm:AGKSY}
Let $X$ be a real random variable. Suppose that $\E{X} = 0$, $\E{X^2} = \sigma^2$, and $\E{X^4} < b\sigma^4$
for some $b > 0$. Then $\Pr\big(X \geq \sigma / (2\sqrt b)\big)> 0$.	
\end{theorem}

\begin{proof}[Proof of Theorem~\ref{thm:main}]
Let $V'$ be the set of variables that $\Phi$ depends on (see Section~\ref{sec:variance-CSP} for definitions).
By Theorem~\ref{thm:variance}, $\Varr{\Phi} \geq |V'| \cdot \beta_k / k$ for some absolute constant $\beta_k > 0$.
By Lemma~\ref{lem:bound-on-C}, the function $\Phi - \Exp \Phi$ satisfies condition (\ref{eq:Bonami-Fourth-Moment})
of the Bonami Lemma for the Efron---Stein Decomposition (Theorem~\ref{thm:bonami-for-ES}) with some absolute constant $C_k$. Hence,
$\|\Phi - \Exp \Phi\|_4^4 \leq 81^k C_k  \|\Phi - \Exp \Phi\|_2^4$. Applying Theorem~\ref{thm:AGKSY}
to the random variable $\Phi - \Exp \Phi$ with $\sigma = (|V'| \beta_k / k)^{1/2}$ and $b = 81^k C_k$,
we get that $\Prob{\Phi \geq \Exp \Phi + |V'|^{1/2}/\kappa_k} > 0$, where $\kappa_k^2 = 4 kC_k 81^k/\beta_k$. Consequently,
$$OPT = \max_{x\in [-1,1]^n} \Phi (x) \geq \Exp \Phi + |V'|^{1/2}/\kappa_k = AVG + |V'|^{1/2}/\kappa_k.$$

We are now ready to state the algorithm. The algorithm computes the Efron---Stein decomposition in time $O_k(m+n)$.
Then, using the formula $V' = \bigcup_{S: \Phi_S \neq 0} \set{x_i:i\in S}$ (see Theorem~\ref{thm:variance}), it
finds the set $V'$ also in time $O_k(m+n)$. It considers two cases.
\begin{enumerate}
\item If $|V'| \geq \kappa_k t^2$, then the algorithm returns $OPT \geq AVG + t$.
\item Otherwise, if $|V'| < \kappa_k t^2$, the algorithm outputs
the restriction of $\calI$ to the variables in $V'$. This is a kernel for $\calI$, since $\Phi$ depends only on the variables in $V'$.
\end{enumerate}
\vspace{-2mm}
\end{proof}
\vspace{-2mm}
To prove Theorem~\ref{thm:main-intro}, we need to show how to find an assignment satisfying $AVG + t$ constraints if $|V'| \geq \kappa_k t^2$.
This can be easily done using $4k$-rankwise independent permutations. A random permutation $\tilde{\alpha}$ is $m$-rankwise independent if
for every subset $M\subset \{1,\dots, n\}$ of size $m$, the order of elements in $M$ induced by $\tilde{\alpha}$ is uniformly distributed (the definition
is due to~Itoh, Takei, and Tarui~\cite{ITT}). Note that any $m$-wise independent permutation $\tilde{\alpha}$ is also an
$m$-rankwise independent permutation. Using the result of Alon and Lovett~\cite{AL}, we can obtain
a $4k$-wise independent permutation $\tilde{\alpha}$ supported on a set of size $n^{O(k)}$. In Lemma~\ref{lem:rankwise} (in Appendix~\ref{sec:rankwise}),
we show that for some permutation $\alpha^*$ in the support of $\tilde{\alpha}$, we have $\val_{\calI}(\alpha^*)\geq AVG +t$.
Hence, to find an assignment satisfying $AVG+t$ constraints, we need to search for the best permutation in the support of $\alpha^*$,
which can be done in time $n^{O(k)}$.

\section{Bonami Lemma} \label{sec:Bonami}
In this section, we prove the Bonami Lemma for the Efron---Stein decomposition (Theorem~\ref{thm:bonami-for-ES}) stated in the introduction.
Our starting point will be the standard Bonami Lemma for Bernoulli $\pm 1$ random variables.
\begin{lemma}[see~\cite{Bonami, ODonnell}]\label{lem:bonami-standard}
Let $f:\{-1,1\}^n \to {\mathbb R}$ be a polynomial of degree at most $d$. Let $X_1,\dots, X_n$ be independent
unbiased $\pm 1$-Bernoulli variables. Then
$$\E{f(X_1,\dots,X_n)^4} \leq 9^{d} \E{f(X_1,\dots,X_n)^2}^2.$$
\end{lemma}

We will consider the following probability distribution in this proof. Let $Z$ be a random variable equal to 3 with probability $1/4$
and to $-1$ with probability $3/4$. Denote by $\cal Z$ the probability distribution of $Z$.
We first prove a variant of the Bonami Lemma for random variables distributed according to $\cal Z$.
\begin{lemma}\label{lem:bonami-Z}
Let $f:\{-1,3\}^n \to {\mathbb R}$ be a polynomial of degree at most $d$. Let $Z_1,\dots, Z_n$ be independent random variables
distributed according to $\cal Z$. Then
$$\E{f(Z_1,\dots, Z_n)^4} \leq 81^{d} \E{f(Z_1,\dots, Z_n)^2}^2.$$
\end{lemma}
\begin{proof}
Consider $2n$ Bernoulli random variables $Y_1',\dots, Y_{n}',Y_1'',\dots, Y_n''$ (uniformly distributed in $\{-1,1\}$).
Note that random variables $\tilde Z_i \equiv Y_i' + Y_i'' + Y_i' Y_i''$ are distributed in the same way as random variables $Z_i$. Therefore,
\begin{align*}
\E{f(Z_1,\dots, Z_n)^4} &= \E{f(\tilde Z_1,\dots, \tilde Z_n)^4},\\
\E{f(Z_1,\dots, Z_n)^2}^2 &= \E{f(\tilde Z_1,\dots, \tilde Z_n)^2}^2.
\end{align*}
Now $f(\tilde Z_1,\dots, \tilde Z_n)$ is a polynomial of $\pm 1$ variables $Y_1',\dots, Y_{n}',Y_1'',\dots, Y_n''$ of degree at most $2d$.
Applying Lemma~\ref{lem:bonami-standard} to $f(\tilde Z_1,\dots, \tilde Z_n)$, we get
$$\E{f(\tilde Z_1,\dots, \tilde Z_n)^4} \leq 9^{2d} \E{f(\tilde Z_1,\dots, \tilde Z_n)^2}^2,$$
and, therefore,
$$\E{f(Z_1,\dots, Z_n)^4} \leq 81^d \E{f(Z_1,\dots, Z_n)^2}^2,$$
as required.
\end{proof}

\noindent Now let $f \in L_2(\Omega^n, \mu^n)$ and $f = \sum_{S} f_S$ be its Efron---Stein decomposition. Define polynomial
$M_f\colon{\mathbb R}^n \to {\mathbb R}$ by
\begin{align*}
M_{f,S} &= (1/3)^{|S|/2} \E{f_S^2}^{1/2} \quad\text{ for every } S\subset \{1,\dots,n\},\\
M_f(Z_1,\dots, Z_n) &= \sum_{S\subset \{1,\dots,n\}} M_{f,S}\prod_{i\in S} Z_i.
\end{align*}
We now get bounds for moments of $f(X_1,\dots, X_n)$ in terms of moments of $M_f(Z_1,\dots,Z_n)$.
\begin{claim}\label{claim:second-moment}
Note that $\E{Z_i} = 0$ and thus $\E{M_f(Z_1,\dots,Z_n)} = M_{f,\varnothing} = |f_{\varnothing}| = |\E{f}|$.

Also,
$$\E{M_f(Z_1,\dots,Z_n)^2} = \sum_{S\subset \{1,\dots,n\}} M_{f,S}^2 \cdot \Exp[\prod_{i\in S}Z_i^2].$$
\end{claim}
\begin{proof}
Note that $\E{Z_i} = 0$ and thus $\E{M_f(Z_1,\dots,Z_n)} = M_{f,\varnothing} = |f_{\varnothing}| = |\E{f}|$.

Also,
$$\E{M_f(Z_1,\dots,Z_n)^2} = \sum_{S\subset \{1,\dots,n\}} M_{f,S}^2 \cdot \Exp[\prod_{i\in S}Z_i^2].$$
Since $\E{Z_i^2} = 3$, we have $\E{\prod_{i\in S}Z_i^2} = 3^{|S|}$, and therefore
$$\E{M_f(Z_1,\dots,Z_n)^2} = \sum_{S\subset \{1,\dots,n\}} \Bigl((1/3)^{|S|/2} \E{f_S^2}^{1/2}\Bigr)^2 \cdot 3^{|S|}= \sum_{S\subset \{1,\dots,n\}} \E{f_S^2} = \E{f^2}.$$
\end{proof}

\begin{claim}\label{claim:fourth-moment}
Let $f$ and $C$ be as in the condition of Theorem~\ref{thm:bonami-for-ES}. Then
$$\E{f(X_1,\dots,X_n)^4} \leq C\, \E{M_f(Z_1,\dots, Z_n)^4}.$$
\end{claim}
\begin{proof}
Write,
\begin{align*}
\E{f^4} &= \sum_{S_1,S_2,S_3,S_4} \E{f_{S_1} f_{S_2} f_{S_3} f_{S_4}},\\
\E{M_f^4} &= \sum_{S_1,S_2,S_3,S_4} M_{f,S_1}M_{f,S_2}M_{f,S_3}M_{f,S_4}
\E{\prod_{i_1\in S_1} Z_{i_1} \prod_{i_2\in S_2} Z_{i_2} \prod_{i_3\in S_3} Z_{i_3} \prod_{i_4\in S_4} Z_{i_4}}.
\end{align*}

To prove the claim, we show that for every four sets $S_1,S_2,S_3,S_4$, the following inequality holds:
\begin{equation}\label{enq:term-by-term}
\E{f_{S_1} f_{S_2} f_{S_3} f_{S_4}} \leq C M_{f,S_1} M_{f,S_2} M_{f,S_3} M_{f,S_4}
\E{\prod_{i_1\in S_1} Z_{i_1} \prod_{i_2\in S_2} Z_{i_2} \prod_{i_3\in S_3} Z_{i_3} \prod_{i_4\in S_4} Z_{i_4}}.
\end{equation}
Note first that if some index $j$ appears in exactly one of the sets $S_1$, $S_2$, $S_3$, and $S_4$ then the expressions on the left and on the right
are equal to $0$ (by Property 3 of the Efron---Stein decomposition in Section~\ref{sec:EfronStein}), and we are done. So we assume that every index $j$ in $S_1 \cup S_2 \cup S_3 \cup S_4$ appears in at least 2 of the sets $S_i$. Denote the number of times $j$ appears in sets $S_i$ by $m(j)$.

Applying the Cauchy---Schwarz inequality three times, we get,
\begin{align}
\E{f_{S_1}f_{S_2}f_{S_3}f_{S_4}} &\leq \E{(f_{S_1}f_{S_2})^2}^{1/2} \E{(f_{S_3}f_{S_4})^2}^{1/2}
= \left(\E{f_{S_1}^2 \cdot f_{S_2}^2} \E{f_{S_3}^2 \cdot f_{S_4}^2}\right)^{1/2}\\
&\leq \left(\E{f_{S_1}^4}^{1/2} \E{f_{S_2}^4}^{1/2} \E{f_{S_3}^4}^{1/2} \E{f_{S_4}^4}^{1/2}\right)^{1/2}\\
&= \|f_{S_1}\|_4 \cdot \|f_{S_2}\|_4 \cdot  \|f_{S_3}\|_4 \cdot \|f_{S_4}\|_4.
\end{align}
By the condition of Theorem~\ref{thm:bonami-for-ES} and the definition of coefficients $M_{f,S_i}$,
\begin{align}\label{eq:f-product}
\E{f_{S_1}f_{S_2}f_{S_3}f_{S_4}} &\leq \|f_{S_1}\|_4 \cdot \|f_{S_2}\|_4 \cdot  \|f_{S_3}\|_4 \cdot \|f_{S_4}\|_4\\
&\leq C \left(\E{f_{S_1}^2}\E{f_{S_2}^2}\E{f_{S_3}^2}\E{f_{S_4}^2}\right)^{1/2} \\
&=
C \cdot 3^{(|S_1| + |S_2| + |S_3| + |S_4|)/2} M_{f,S_1} M_{f,S_2} M_{f,S_3} M_{f,S_4}. \notag
\end{align}
We also have,
\begin{align*}
\E{\prod_{i_1\in S_1} Z_{i_1} \prod_{i_2\in S_2} Z_{i_2} \prod_{i_3\in S_3} Z_{i_3} \prod_{i_4\in S_4} Z_{i_4}}  &=
\E{\prod_{r=1}^4\prod_{i\in S_r} Z_i}
= \prod_{i\in \bigcup_r S_r}
\E{Z_i^{m(i)}} \\
&=
\prod_{r=1}^4\prod_{i\in S_r} \Big(\E{Z_i^{m(i)}}\Big)^{1/m(i)}.
\end{align*}
We compute $\E{Z_i^{m(i)}}$ for $m(i)\in\{2,3, 4\}$. We get $\Exp Z_i^2 = 3$, $\Exp Z_i^3 = 6$ and $\Exp Z_i^4 = 21$.
Thus, $(\Exp Z_i^2)^{1/2} = \sqrt{3}$, $(\Exp Z_i^3)^{1/3} = \sqrt[3]{6}> \sqrt{3}$ and $(\Exp Z_i^4)^{1/4} = \sqrt[4]{21} > \sqrt{3}$,
and, consequently,
\begin{equation}\label{eq:Z-product}
\E{\prod_{r=1}^4\prod_{i\in S_r} Z_i} \geq \prod_{r=1}^4 \prod_{i\in S_r} 3^{1/2} = 3^{(|S_1| + |S_2| + |S_3| + |S_4|)/2}.
\end{equation}
Since all coefficients $M_{f,S}$ are non-negative, we get from (\ref{eq:f-product}) and (\ref{eq:Z-product})  that  inequality (\ref{enq:term-by-term}) holds.
\end{proof}

\begin{proof}[Proof of Theorem~\ref{thm:bonami-for-ES}]
By Lemma~\ref{lem:bonami-Z}, we have
$$\|M_f(Z_1,\dots,Z_n)\|_4^4 \leq 81^d \|M_f(Z_1,\dots,Z_n)\|_2^4.$$
From Claims~\ref{claim:second-moment} and~\ref{claim:fourth-moment}, we get
$$\E{f(X_1,\dots,X_n)^4} \leq C\E{M_f(Z_1,\dots,Z_n)^4} \leq 81^d C \E{M_f(Z_1,\dots,Z_n)^2}^2  =
81^d C \E{f(X_1,\dots, X_n)^2}^2.$$
\end{proof}

\section{General Framework}\label{sec:general-framework}
\subsection{Filtered $A$-Lattice of Functions}
In this section, we generalize the result of the paper to a more general class of constraint satisfaction problems having a lattice structure. In Section~\ref{sec:PPP},
we show that LP CSPs and valued CSPs with ``piece-wise polynomial predicates'' (see Section~\ref{sec:PPP} for the defintion) have a lattice structure.

\subsection{Discussion}
We note that in our proofs we used only few properties of ordering CSPs. Specifically, in Theorem~\ref{thm:ES-ordering-predicate}, we showed that
all functions in the Efron---Stein decomposition of the basic ordering predicate are in the set
$$
\calF^{\mathrm{ord}}_{d} = \set{\sum_{\tau}\phi_{\tau}(x_1,\dots, x_k) \frac{q_{S,\tau}(x_1,\dots, x_n)}{2^{d} d!}}.
$$
Since $\calF^{\mathrm{ord}}_{d}$ is closed under addition (the sum of any two functions in $\calF^{\mathrm{ord}}_{d}$ is in $\calF^{\mathrm{ord}}_{d}$), we got that
all functions in the Efron---Stein decomposition of the ordering CSP objective $\Phi$ are also in $\calF^{\mathrm{ord}}_{d}$.
Then in the proof of  Lemma~\ref{lem:beta}, we showed that every non-zero function in $\calF^{\mathrm{ord}}_{d}$  has
variance at least $\beta_k$ (where $\beta_k$ depends only on $k$), and this was sufficient to get the result of the paper.
To summarize, we only used the following properties of the set of functions $\calF^{\mathrm{ord}}_{d}$:
\begin{itemize}
\item[A.] all functions in the Efron---Stein decomposition of each predicate are in $\calF^{\mathrm{ord}}_{d}$,
\item[B.] $\calF^{\mathrm{ord}}_{d}$ is closed under addition,
\item[C.] every non-zero function in $\calF^{\mathrm{ord}}_{d}$ has variance at least $\beta$ (for some fixed $\beta > 0$).
\end{itemize}
\subsection{Filtered $A$-Lattice of Functions}
We now formalize properties A, B, and C in the definitions of $A$-lattice of functions and filtered $A$-lattice of functions.
Recall first the definition of a lattice.
\begin{definition}
Let $V$ be a finite-dimensional space and $\calL$ be a subset of $V$. We say that $\calL$ is a lattice in $V$ if
for some basis $v_1,\dots, v_r$ of $V$, we have $\calL = \set{\sum_{i=1}^r a_i v_i:a_1,\dots, a_r \in \bbZ}$.
We say that $v_1,\dots, v_r$ is the basis of the lattice $\calL$.
\end{definition}

Now we define an $A$-lattice of functions.
\begin{definition}\label{def:A-lattice}
Let $(\Omega,\mu)$ be a probability space.
Consider a set $\calF$ of bounded (real-valued) functions $f(x_1,\dots, x_k)$  on $\Omega$; $\calF \subset L_{\infty}(\Omega^k)$.
We say that $\calF$ is an $A$-lattice of functions of arity (at most) $k$ on $\Omega$ if it satisfies the following properties.
\begin{enumerate}
\item $\calF$ is a lattice in a finite dimensional subspace of $L_{\infty}(\Omega^k)$.
\item If we permute arguments of a function in $\calF$, we get a function in $\calF$. Specifically, if $f\in \calF$ and $\pi$ is a permutation of $\set{1,\dots, k}$ then $g(x_1,\dots, x_k) = f(x_{\pi(1)},\dots, x_{\pi(k)}) \in\calF$.
\end{enumerate}

For an $A$-lattice $\calF$, we write that a function $f\in_R \calF$ if $f$ is in $\calF$ after possibly renaming the arguments of $f$
(in other words, $f$ is in $\calF$ as an abstract function from $\Omega^k$ to $\bbR$).\footnote{For example, let $\calF$ be an $A$-lattice of
functions of the form $ax_1 + bx_2$ where $a,b\in \bbZ$. Then $x_3 + 5 x_7 \in_R \calF$, since after renaming $x_3$ to $x_1$ and $x_7$ to $x_2$ we get
$x_3 + 5x_2$, which is of the form  $ax_1 + bx_2$.}
\end{definition}
Clearly, every $A$-lattice $\calF$ of functions satisfies property~B. Since $\calF$ is discrete, it also satisfies property~C (we will prove that formally in
Claim~\ref{claim:compact-A-CSP}). We also want to ensure that it satisfies an analog of property~A.
To this end, we consider the averaging operator $A_i$, which takes the expectation of a function with respect to variable $x_i$  and require
that $A_i$ maps every function in the lattice to a function in the lattice.

\begin{definition}
For $i\in \set{1,\dots, k}$, let $A_i$ be the averaging operator that maps a function $f$ of arity $k$ to a function $A_i f$ of arity $k-1$ defined as follows:
\begin{align*}
A_i f(x_1,\dots,&\, x_{i-1},x_{i+1},\dots, x_k) = \int_{\Omega} f(x_1,\dots, x_k) \,d\mu(x_i)\\
&= \Exp{f(X_1,\dots, X_k|X_1 = x_1,\dots, X_{i-1} = x_{i-1},X_{i+1} = x_{i+1},\dots, X_k = x_k)}.
\end{align*}
\end{definition}
\begin{definition}
We say that a family of sets $\{\calF_{\alpha}\}_{\alpha}$ (indexed by integer $\alpha\geq 1$) is a filtered $A$-lattice of
functions of arity (at most) $k$ if it satisfies the following properties.
\begin{enumerate}
\item $\calF_{\alpha}$ is an $A$-lattice of functions of arity $k$ on $\Omega$.
\item $\set{\calF_{\alpha}}$ is a filtration: $\calF_{\alpha} \subset \calF_{\alpha'}$ for $\alpha \leq \alpha'$.
\item For every $\alpha$ there exists $\alpha'$, which we denote by $\alpha' = a(\alpha)$, such that the operator $A_i$ maps $\calF_{\alpha}$ to $\calF_{\alpha'}$ (for every $i\in \set{1,\dots, k}$).
\end{enumerate}
\end{definition}
We remark  that $\{\calF^{\mathrm{ord}}_{d}\}$ is a filtered $A$-lattice.
We are going to prove that our result for ordering CSPs holds, in fact, for any constraint satisfaction problem with predicates from a filtered $A$-lattice.

\subsection{General $A$-CSP($\calF_{\alpha}$, $\alpha$)}
\begin{definition}
Consider a probability space $(\Omega, \mu)$. Let $\calF_{\alpha}$ be a filtered $A$-lattice of functions and $\alpha_0$ is an integer.
An instance $\calI$ of General $A$-CSP($\calF_{\alpha}$,  $\alpha_0$) consists of a set of variables $x_1,\dots, x_n$, taking values in $\Omega$,
and a set of real-valued constraints of the form $f(x_{i_1}, \dots, x_{i_k})$  where $f\in_R \calF_{\alpha}$. The objective function $\Phi(x_1,\dots, x_n)$ is the sum of
all the constraints.
General $A$-CSP($\calF_{\alpha}$, $\alpha_0$) asks to find an assignment to variables $x_1, \dots, x_n$ that maximizes
$\Phi(x_1,\dots, x_n)$.
\end{definition}
We denote the optimal value of an instance $\cal I$ by $OPT = \esssup_{x_1,\dots, x_n \in \Omega} \Phi(x_1,\dots, x_n)$ and the average
by $AVG = \E{\Phi(X_1,\dots, X_n)}$, where $X_1,\dots, X_n$ are independent random elements of $\Omega$ distributed according to the probability measure $\mu$.
\begin{remark}\label{rem:measure-0}
Note that we follow the standard convention that two functions $f,g \in L_{\infty}(\Omega^k)$ are equal if $f(x_1,\dots, x_k) \neq g(x_1,\dots,x_k)$ on a set of measure $0$.
That is, we identify functions that are equal almost everywhere. Accordingly, we define $OPT$ as the essential supremum  of $\Phi$:
 $OPT$ is equal to the maximum value of $M$ such that
$$\Prob{\Phi \geq M - \varepsilon} > 0,$$
for every $\varepsilon > 0$.
\end{remark}
We prove a counterpart of Theorem~\ref{thm:main} for the General $A$-CSP problem.

\begin{theorem}\label{thm:main-A-CSP}
There is an algorithm that given an instance of General $A$-CSP($\calF_{\alpha}$, $k$, $\alpha_0$) and a parameter $t$,
either finds a kernel on at most $\kappa t^2$ variables (where $\kappa$ depends only on the filtered $A$-lattice $\set{\calF_{\alpha}}$ and numbers $k, \alpha_0$)
or
certifies that $OPT \geq  AVG + t$.
The algorithm runs in time $O(m+n)$, linear in the number of constraints $m$ and variables $n$ (the coefficient in the $O$-notation depends the filtered $A$-lattice $\set{\calF_{\alpha}}$ and numbers $k, \alpha_0$).

 We assume that computing the sum of two functions in  $\calF_{\alpha}$ requires constant time and that computing $A_i f$ requires constant time
(the time may depend on $\alpha$).
\end{theorem}

We first prove analogues of Theorem~\ref{thm:ES-ordering-predicate} and Corollary~\ref{cor:ES-objective} for General $A$-CSP.
\begin{lemma}[cf. Theorem~\ref{thm:ES-ordering-predicate}]\label{lem:ES-general-predicate}
Consider a probability space $(\Omega,\mu)$ and a filtered $A$-lattice $\set{\calF_{\alpha}}$.
For every $\alpha$, there exists $\tilde\alpha$ so that the following holds. For every $f \in \calF_{\alpha}$ and every $S\subset \set{1,\dots, k}$,
$f_S \in \calF_{\tilde\alpha}$.
\end{lemma}
\begin{proof}
Let $\alpha_0 = \alpha$, $\alpha_1 = a(\alpha_0)$ (where $a$ is as in the definition of a filtered $A$-lattice), $\alpha_2 = a(\alpha_1)$, and so on;
$\alpha_i = a(\alpha_{i-1})$. Let $\tilde \alpha = \max_{i\in \set{0,\dots, k}} \alpha_i$.
Consider a function $f \in\calF_{\alpha}$ and a set $T\subset\set{1,\dots,k}$. By~(\ref{eq:ES-1}),
$$f_{\subset T} = \E{f(X_1,\dots, X_n)| \text{ all } X_i \text{ with } i \in T}.$$
Denote the elements of $\set{1,\dots,k}\setminus T$ by $i_1< \dots < i_t$ (where $t = k - |T|$). Note that
$$f_{\subset T} = A_{i_1} A_{i_2} \dots A_{i_t} f \in \calF_{\alpha_{t}} \subset \calF_{\tilde\alpha},$$
by the definition of a filtered $A$-lattice.
Now by~(\ref{eq:ES-2}),
$$f_S = \sum_{T\subset S} (-1)^{|S\setminus T|} f_{\subset T}.$$
Since $\calF_{\tilde\alpha}$ is a lattice and $f_S$ is a linear combination, with integer coefficients, of functions $f_{\subset T}$ (all of which are in $\calF_{\tilde\alpha}$),
$f_S$ is in $\calF_{\tilde\alpha}$.
\end{proof}
Since the set of functions $\calF_{\tilde\alpha}$ is closed under addition, we get that for every General $A$-CSP($\calF_{\alpha}$, $\alpha$) instance $\calI$ with objective
function $\Phi$, all functions $\Phi_S$ are also in $\calF_{\tilde\alpha}$.
\begin{corollary}[cf. Corollary~\ref{cor:ES-objective}]
Consider a probability space $(\Omega,\mu)$ and a filtered $A$-lattice $\set{\calF_{\alpha}}$.
Let $\calI$ be an instance of General $A$-CSP($\calF_{\alpha}$, $\alpha_0$) and let $\Phi$ be its objective functions.
Then for every subset $S\subset \set{1,\dots, n}$ of size at most $k$, $\Phi_S\in_R\calF_{\tilde\alpha}$;
for every subset $S$ of size greater than $k$, $\Phi_S = 0$.

Furthermore, the Efron---Stein decomposition $\set{\Phi_S}$ of $\Phi$ can be computed in time $O(m)$.
\end{corollary}
\subsection{Compactness Properties of Filtered $A$-Lattices}
We now prove counterparts of Claim~\ref{claim:compact} and Lemma~\ref{lem:bound-on-C} for General A-CSP.
\begin{claim}[cf. Claim~\ref{claim:compact}]\label{claim:compact-A-CSP}
Consider a probability space $(\Omega,\mu)$. Let $\calF$ be an $A$-lattice of functions of arity $k$ on $\Omega$.
There exists a positive number $\beta$ such that for every function $f\in\calF$, $\E{f^2}=\E{f(X_1,\dots, X_k)^2} \geq \beta$.
\end{claim}
\begin{proof}
Let $f_1,\dots,f_r$ be the basis of lattice $\calF$. Consider the linear span $Q$ of functions $f_1,\dots, f_r$ (the set of all linear combinations with real coefficients).
Vector space $Q$ is finite dimensional. Let $Q_1 = \set{\sum a_i f_i: \max_i |a_i| = 1}$. Note that $Q_1$ is a compact set. All functions in $Q_1$ are non-zero (since
$f_1,\dots, f_r$ are linearly independent),
and, therefore, $\E{g^2} > 0$ for every $g\in Q_1$. Since $Q_1$ is compact,  $\min_{g\in Q_1} \E{g^2} = \inf_{g\in Q_1} \E{g^2} > 0$. Denote $\beta = \min_{g\in Q_1} \E{g^2}$.

Now consider a non-zero function $f\in\calF$. Write $f = \sum_{i=1}^r a_i f_i$. Let $M = \max_i |a_i|$. Since $f \neq 0$ and all coefficients $a_i$ are integer, $M \geq 1$.
Note that $f/M \in Q_1$. We have,
$$\E{f^2} = M^2\, \E{(f/M)^2} \geq \beta M^2 \geq \beta,$$
as required.
\end{proof}

\begin{lemma}[cf.~Lemma~\ref{lem:bound-on-C}]\label{lem:bound-on-C-A-CSP}
Consider a probability space $(\Omega,\mu)$. Let  $\calF$ be an $A$-lattice of functions of arity $k$ on $\Omega$.
There exists a constant $C$ such that the following holds. Let $f$ be a functions of arity at most $k$, which
depends on a subset of variables $\set{x_1,\dots, x_n}$. Assume that $f\in_R \calF$ (see Definition~\ref{def:A-lattice}). Then
$$\E{f^4} \leq C\,  \E{f^2}^{2}.$$
\end{lemma}
\begin{proof}
Since function $f$ depends on at most $k$ variables among $x_1,\dots, x_n$, we may assume without loss of generality
that it depends on a subset of $x_1,\dots, x_{k}$. Then $f\in \calF$.
Let $\calQ$ be the linear span of $\calF$.
Note that $\calQ$ is a finite dimensional vector space of functions.

Define $\calQ_1 =\set{h \in \calQ: \|h\|_2 = 1}$. Since $\calQ_1$ is
a compact set, the continuous
function
$$W(g) = \E{g^4}$$
is  bounded when $g \in\calQ_1$. Denote its maximum by $C$.

Letting $g = f / \|f\|_2$, we have,
$$\E{f^4} =
\E{g^4} \cdot \|f\|_2^4
 \leq C  \|f\|_2^4,$$
as required.
\end{proof}

\subsection{Variance of $A$-CSP Objective}
We now prove a counterpart of Theorem~\ref{thm:variance} for General $A$-CSP($\calF_{\alpha}$, $k$, $\alpha$).
\begin{lemma}[cf. Theorem~\ref{thm:variance}]\label{lem:variance-A-CSP}
Consider a probability space $(\Omega,\mu)$. Let  $\set{\calF_{\alpha}}$ be a filtered $A$-lattice and $\alpha_0$ be an integer.
There exists a number $\beta > 0$, which depends only on $\set{\calF_{\alpha}}$ and $\alpha_0$ such that the following holds.
Let $\cal I$ be an instance of General $A$-CSP($\calF_{\alpha}$, $\alpha_0$) and  $\sigma$ be a parameter.
Either $\cal I$ has a kernel on at most $(k / \beta) \sigma^2$ variables or $\Varr{\Phi} \geq \sigma^2$. Moreover,
there is an algorithm that either finds a kernel on at most $(k / \beta) \sigma^2$ variables or certifies that
$\Varr{\Phi} \geq \sigma^2$. The algorithm runs in time $O(m+n)$, where $n$ is the number of variables and $m$ is the number of constraints.
\end{lemma}
\begin{proof}
Let $\tilde\alpha$ be as in Lemma~\ref{lem:ES-general-predicate}. Since $\calF_{\tilde\alpha}$ is an $A$-lattice, by Claim~\ref{claim:compact-A-CSP},
there exists $\beta > 0$ such that $\|f\|^2_2 \geq \beta$ for every non-zero $f\in \calF_{\tilde\alpha}$.
Note that $\beta$ does not depend on $n$ and $t$.

Consider the Efron---Stein decomposition of $\Phi$. Let $V' = \bigcup_{S: \Phi_S \neq 0} \set{x_i:i\in S}$.
Function $\Phi$ depends only on variables in $V'$. Therefore, the restriction of $\cal I$ to variables in $V'$ is a kernel for $\cal I$. Let $\nu = |V'|$.
If $\nu < (k / \beta) \sigma^2$, then we are done. So let us assume that $\nu \geq (k / \beta) \sigma^2$.
There are at least $\nu / k$ non-empty sets $S$ with $\Phi_S \neq 0$ since each such set $S$ contributes at most $k$ variables to $V'$.
Note that $\E{\Phi_S} = 0$ for $S\neq \varnothing$ and hence $\Varr{\Phi_S} = \E{\Phi_S^2}$. Since $\Phi_S \in  \calF_{\tilde\alpha}$,
$\Varr{\Phi_S} = \E{\Phi_S^2} \geq \beta$, if $\Phi_S \neq 0$ and $S\neq \varnothing$. We have,
$$\Varr{\Phi} = \sum_S \Varr{\Phi_S}
\geq |\set{S\neq \varnothing:\Phi_S \neq 0}| \beta
\geq (\nu / k) \beta \geq \sigma^2,
$$
as required.

Note that we can compute the Efron---Stein decomposition of $\Phi$ in time $O(m+n)$ and then find the set $V'$ in time $O(m+n)$. If $|V'| < (k/\beta)\sigma^2$, we output the restriction of $\cal I$ to $V'$ (which we compute in time $O(m+n)$). Otherwise, we output that $\Varr{\Phi} \geq \sigma^2$.
\end{proof}

\subsection{Proof of Theorem~\ref{thm:main-A-CSP}}
We are ready to prove Theorem~\ref{thm:main-A-CSP}.
\smallskip
\begin{proof}
Let $\beta$ be as in Lemma~\ref{lem:variance-A-CSP} and $C$ be as in Lemma~\ref{lem:bound-on-C-A-CSP}.
Let $\sigma^2 = 4 \cdot 81^k C t^2$. Denote $$f= \Phi - \E{\Phi} = \Phi - AVG.$$ Note that $\esssup_{x\in\Omega^n} f(x_1,\dots, x_n) = MAX - AVG$.

By Lemma~\ref{lem:variance-A-CSP}, either $\cal I$ has a kernel on at most $(k/\beta) \sigma^2 = (4 \cdot 81^k \cdot k C/\beta) t^2 $ variables or
$\Varr{f} \geq \sigma^2$. In the former case, we output the kernel, and we are done. In the latter case, we show that $OPT - AVG \geq t$. Assume that $\Varr{f} \geq \sigma^2$. By
Theorem~\ref{thm:bonami-for-ES} and Lemma~\ref{lem:bound-on-C-A-CSP},
$\|f\|_4^4 \leq 81^k C \|f\|_2^2$.
By Theorem~\ref{thm:AGKSY}, $\Prob{f \geq \sigma / (2 \cdot 9^k \sqrt{C})} > 0$, and hence $MAX - AVG \geq \sigma / (2 \cdot 9^k \sqrt{C}) \geq t$.

The algorithm only executes the algorithm from Lemma~\ref{lem:variance-A-CSP}, so its running time is $O(m+n)$.
\end{proof}

\section{Piecewise Polynomial Predicates}\label{sec:PPP}
In this section, we present an interesting example of a filtered $A$-lattice, the set of piecewise polynomial functions.  As a corollary,
we get that  the problem of maximizing the objective over average for a CSP with piecewise polynomial functions is fixed-parameter tractable.

\begin{definition}
Let us say that a subset $P$ of $[-1,1]^k$ is $b$-polyhedral if it is defined by a set of linear inequities on $x_1, ..., x_k$, in which all coefficients are bounded by $b$ in absolute value.
 In other words, $P$ is a $b$-polyhedral set if for some $t$ there exist a $k \times t$ matrix $A$ and vector $c$ (with $t$ coordinates) such that
  $P=\set{x:Ax < c}$ (here, the inequality $Ax <c$ is understood coordinate-wise), and  every entry of $A$ and coordinate of $c$ is bounded by $b$ in absolute value.
 We denote the indicator function of a polyhedral set $P$ by $I_P$.
\end{definition}
\begin{definition}
We denote the set of polynomials $f(x_1,\dots, x_k)$ with real coefficients  of degree at most $d$ by $\bbR_{\leq d}[x_1,\dots,x_k]$;
we denote the set of polynomials  $f(x_1,\dots, x_k)$ with integer coefficients of degree at most $d$  by $\bbZ_{\leq d}[x_1,\dots,x_k]$.
\end{definition}
\begin{definition}
We say that a function $f(x_1,..., x_k): [-1,1]^k \to \bbR$ is piecewise polynomial on polyhedral sets or $(d,b)$-PPP if $f$ is the sum of terms of the form $g(x_1,..., x_k) I_P(x_1,...,x_k)$, where $g\in \bbZ_{\leq d}[x_1,\dots, x_k]$ and $P$ is a $b$-polyhedral set.
\end{definition}

We note that every $(d,b)$-PPP function can be written in the following ``canonical form''. Consider all hyperplanes in $\bbR^k$ of the form $\langle a, x\rangle = c$,
in which $a,c \in\set{-b,\dots,b}^k$. They partition $[-1,1]^d$ into polyhedrons. We call these polyhedrons elementary polyhedrons and denote the set of all
elementary polyhedrons by $\calP_{elem}$. Note that each $b$-polyhedral set is a union of elementary polyhedrons.
Thus we can write every $(d,b)$-PPP function $f$ as follows:
\begin{equation}
f(x_1,\dots, x_k) = \sum_{P\in \calP_{elem}} I_P(x_1,\dots,x_k) g_P(x_1,\dots,x_k), \label{eq:canonical-form}
\end{equation}
where $g_P \in \bbZ_{\leq d}[x_1,\dots,x_k]$.

\begin{theorem}\label{thm:PPP-is-FAL}
Let $\Omega = [-1,1]$ and $\mu$ be the uniform measure on $[-1,1]$.
Let
$$\calF_{\alpha} = \set{f: (\alpha! f) \text{ is an } (\alpha,\alpha)\text{-PPP function of variables } x_1,\dots, x_k}.$$
Then $\calF_{\alpha}$ is a filtered $A$-lattice of functions.
\end{theorem}
\begin{proof}
First, we prove  that each set $\calF_{\alpha}$ is an $A$-lattice.
It follows from (\ref{eq:canonical-form}) that $\calF_{\alpha}$ is a lattice with basis $I_P(x_1,\dots,x_k) g(x_1,\dots,x_k)/\alpha!$,
where $P\in\calP_{elem}$ and $g(x_1,\dots,x_k)$ is a monomial of degree at most $\alpha$ (i.e., $g$ is of the form $x_1^{r_1} \dots x_k^{r_k}$).
Since every monomial $g$ is bounded on $[-1,1]^k$, every basis function is bounded, and, therefore, all functions in $\calF_{\alpha}$ are bounded.
The definition of $\calF_{\alpha}$ is symmetric with respect to $x_1,\dots, x_k$, hence if we permute the arguments of any function $f\in \calF_{\alpha}$,
we get a function in $\calF_{\alpha}$.

Now we show that $\calF_{\alpha}$ is a filtered $A$-lattice. The inclusion $\calF_{\alpha} \subset \calF_{\alpha'}$ for $\alpha\leq \alpha'$ is immediate.
It remains to show that for every $\alpha$ there exists $\alpha'$ such that $A_i$ maps $\calF_{\alpha}$ to $\calF_{\alpha'}$.
Let $a  = \binom{2 \alpha}{\alpha}$ and $\alpha' = a^{\alpha+1}$ (we note that, in fact, we can choose a much smaller value of $\alpha'$;
 however, we use this value to simplify the exposition). Observe  that all integer numbers between $1$ and $\alpha$ divide $a$.

It is sufficient to prove that $A_i$ sends every basis function
$$f(x_1,\dots, x_k) = I_P(x_1,\dots,x_k) g(x_1,\dots,x_k)/\alpha!$$ to $\calF_{\alpha'}$.
Moreover, since the set of functions $\calF_{\alpha}$ is invariant under permutation of function arguments, we may assume without loss of generality that $i = k$.
Denote $g = x_1^{r_1} \dots x_k^{r_k}$, where $r_1 + \dots r_k \leq \alpha$.
Consider the set of linear inequalities $L$ that define polyhedron $P$. All coefficients in each of the inequalities are bounded by $\alpha$ in absolute value.
Let $L_0$ be those inequalities that do not depend on $x_k$ and $L_1$ be those that do depend on $x_k$.
We rewrite every inequality in $L_1$ as follows. Consider an inequality in $L_1$. Let $\lambda$ be the coefficient of $x_k$ in it.
We multiply the inequality by $a/\lambda \in \bbZ$, and if $\lambda < 0$, we change the comparison sign in the inequality to the opposite.
Finally, we move all terms in the inequality other than $a x_k$ to the right hand side.
We get an equivalent inequality of the form either $a x_k > l(1,x_1,\dots, x_{k-1})$ or $a x_k < u(1,x_1,\dots, x_{k-1})$,
where $l$ and $u$ are linear functions with integer coefficients bounded by $\alpha a$ in absolute value.
Denote the inequalities of the form $a x_k > l(1,x_1,\dots, x_{k-1})$ by
$$a x_k > l_1(1,x_1,\dots, x_{k-1}), a x_k > l_2(1,x_1,\dots, x_{k-1}),\dots, a x_k > l_p(1,x_1,\dots, x_{k-1}),$$
and the inequalities of the form $a x_k < u(1,x_1,\dots, x_{k-1})$ by
$$a x_k < u_1(1,x_1,\dots, x_{k-1}), a x_k < u_2(1,x_1,\dots, x_{k-1}),\dots, a x_k < u_q(1,x_1,\dots, x_{k-1}).$$

Let $M_0$ be the set of points $x = (x_1,\dots, x_k)$ such that
$l_{j'} (x_1,\dots, x_k) = l_{j''} (x_1,\dots, x_k)$ or $u_{j'} (x_1,\dots, x_k) = u_{j''} (x_1,\dots, x_k)$
for some $j' \neq j''$. Note that $M_0$ has measure $0$.

Define $pq$ polyhedrons $P_{j_1j_2}$ in $\bbR^{k-1}$. For $j_1\in\set{1,\dots,p}$ and $j_2\in\set{1,\dots,q}$, let $P_{j_1j_2}$ be the
polyhedron defined by the following inequalities:
\begin{enumerate}
\item all inequalities in $L_0$,
\item inequality $l_{j_1}(1,x_1,\dots,x_{k-1}) < u_{j_2}(1,x_1,\dots,x_{k-1})$,
\item inequalities $l_j(1,x_1,\dots,x_{k-1}) < l_{j_1}(1,x_1,\dots,x_{k-1})$ for every $j\neq j_1$,
\item inequalities $u_j(1,x_1,\dots,x_{k-1}) > u_{j_2}(1,x_1,\dots,x_{k-1})$ for every $j\neq j_2$.
\end{enumerate}
Inequalities in items 2--4 are equivalent to the following condition (except for points in $M_0$):
\begin{equation}\label{ineq:Pjj-constraints}
\max_j l_j(1,x_1,\dots,x_{k-1}) = l_{j_1}(1,x_1,\dots,x_{k-1}) < u_{j_2}(1,x_1,\dots,x_{k-1}) = \min_j u_{j}(1,x_1,\dots,x_{k-1}).
\end{equation}

Note that if $(x_1,\dots, x_k)\in P\setminus M_0$ then $(x_1,\dots,x_{k-1}) \in P_{j_1j_2}$ for
$$j_1 = \argmax_j  l_{j}(1,x_1,\dots,x_{k-1}) \text{ and } j_2 = \argmin_j  u_{j}(1,x_1,\dots,x_{k-1}).$$
Also note that polyhedrons $P_{j_1j_2}$ are disjoint.
Now let
$$h_{j_1j_2}(x_1,\dots, x_{k-1}) = \frac{1}{2\alpha! a^{r_k+1} (r_k+1)} (u_{j_2}(1,x_1,\dots,x_{k-1})^{r_k+1} - l_{j_1}(1,x_1,\dots,x_{k-1})^{r_k+1}) x_1^{r_1} \dots x_{k-1}^{r_{k-1}}.$$

Let $x' = (x_1,\dots, x_{k-1})$ be a point in $[-1,1]^{k-1}$. Consider two cases.
\medskip

\paragraph{Case 1} First, assume that $x'\in P_{j_1j_2}$ for some $j_1$ and $j_2$. Then
$$A_kf(x') = \frac{1}{2\alpha!}\int_{-1}^1 I_P(x_1,\dots,x_k) g(x_1,\dots,x_k) dx_k.$$
The point $(x_1,\dots, x_k)$ satisfies all inequalities in $L_0$ since $x'\in P_{j_1j_2}$.
Hence, $I_P(x_1,\dots,x_k) = 1$ if and only if it satisfies all inequalities in $L_1$, which are equivalent to
$$\max_j l_j(1,x_1,\dots,x_{k-1}) \leq ax_k \leq \min_j u_j(1,x_1,\dots,x_{k-1}).$$
Combining this with~(\ref{ineq:Pjj-constraints}), we get that $I_P(x_1,\dots,x_{k-1}) = 1$ if and only if
$$x_k \in \left[\frac{1}{a}l_{j_1}(1,x_1,\dots,x_{k-1}) , \frac{1}{a} u_{j_2}(1,x_1,\dots,x_{k-1})\right].$$
Therefore,
\begin{align*}
A_kf(x') &= \frac{1}{2\alpha!}\int_{l_{j_1}(1,x_1,\dots,x_{k-1}) /a}^{u_{j_2}(1,x_1,\dots,x_{k-1}) /a}  x_1^{r_1} \dots x_k^{r_k} dx_k\\
              &= \frac{1}{2\alpha! a^{r_k+1} (r_k+1)} (u_{j_2}(1,x_1,\dots,x_{k-1})^{r_k+1} - l_{j_1}(1,x_1,\dots,x_{k-1})^{r_k+1}) x_1^{r_1} \dots x_{k-1}^{r_{k-1}}\\
              &= h_{j_1j_2}(x').
\end{align*}

\paragraph{Case 2} Now assume that $x'\notin P_{j_1j_2}$ for every $j_1$ and $j_2$.  Then there is no $x_k$ such that $(x_1,\dots,x_k) \in P\setminus M_0$. Therefore,
$$A_kf(x') = \frac{1}{2\alpha!}\int_{-1}^1 I_P(x_1,\dots,x_k) g(x_1,\dots,x_k) dx_k = \frac{1}{2\alpha!} \int_{-1}^1 0 \cdot g(x_1,\dots,x_k)dx_k = 0.$$
(The equality holds on a set of full measure; see Remark~\ref{rem:measure-0}.)

We conclude that
$$A_kf(x') = \sum_{j_1,j_2} I_{P_{j_1j_2}} h_{j_1j_2}(x').$$
All coefficients in the inequalities that define $P_{j_1j_2}$ are bounded by $2\alpha a$ in absolute value, and
$2\alpha! a^{r_k+1} (r_k+1) h(x') \in \bbZ_{\leq \alpha + 1}$. Therefore,
$2\alpha! a^{r_k+1} (r_k+1) A_kf$ is an $(\alpha+1,2\alpha a)$-PPP function. Thus $A_k f\in\calF_{\alpha'}$.
\end{proof}

From Theorems~\ref{thm:main-A-CSP} and~\ref{thm:PPP-is-FAL}, we get the following corollary.
\begin{corollary}
For every $k$, $d$ and $b$, there is an algorithm that given an instance of a constraint satisfaction problem on $n$ variables $x_1,\dots,x_n$
with $m$ real-valued constraints, each of which is a $(d,b)$-PPP function of arity $k$, and a parameter $t$,
either finds a kernel on at most $\kappa t^2$ variables or certifies that $OPT \geq  AVG + t$.
The algorithm runs in time $O(m+n)$. (The coefficient $\kappa$ and the coefficient in the $O$-notation depend only on $k$, $d$, and $b$).
\end{corollary}
\begin{proof}
Let $\alpha = \max(d, b)$.
We apply Theorem~\ref{thm:main-A-CSP} to filtered $A$-lattice $\calF_{\alpha}$ from Theorem~\ref{thm:PPP-is-FAL} and get the corollary.
\end{proof}

\smallskip
Since every constraint in a  $(k,b)$-LP CSP problem is a $(0,b)$-PPP function of arity $k$ (see Definition~\ref{def:LP-CSP}), we get the following corollary.
\begin{corollary}
For every $k$ and $b$, there is an algorithm that given an instance of $(k,b)$-LP CSP
either finds a kernel on at most $\kappa t^2$ variables
or certifies that $OPT \geq  AVG + t$.
The algorithm runs in time $O(m+n)$ (The coefficient $\kappa$ and the coefficient in the $O$-notation depend only on $k$ and $b$).
\end{corollary}
\begin{remark}
Note that for an instance of $(k,b)$-LP CSP, we have
$$OPT = \max_{x_1,\dots,x_n\in[-1,1]} \Phi(x_1,\dots, x_n) =  \esssup_{x_1,\dots,x_n\in[-1,1]} \Phi(x_1,\dots, x_n),$$
since all LP constraints are strict. If we were to use non-strict ``less-than-or-equal-to'' and ``greater-than-or-equal-to'' LP constraints, we would
have to define
$OPT$ as  $\esssup_{x_1,\dots,x_n\in[-1,1]} \Phi(x_1,\dots, x_n)$, and not as $\max_{x_1,\dots,x_n\in[-1,1]} \Phi(x_1,\dots, x_n)$,
since, in general, $\esssup_{x_1,\dots,x_n\in[-1,1]} \Phi(x_1,\dots, x_n)$ might not be equal to $\max_{x_1,\dots,x_n\in[-1,1]} \Phi(x_1,\dots, x_n)$.
For example, consider an instance of $(2,1)$-LP CSP with two constraints $x_1 \leq x_2$ and $x_2 \leq x_1$; we have
$$OPT = \esssup_{x_1,\dots,x_n\in[-1,1]} \Phi(x_1,\dots, x_n) = 1,$$
but
$$\max_{x_1,\dots,x_n\in[-1,1]} \Phi(x_1,\dots, x_n) = 2.$$
(The maximum is attained on a set of measure $0$, where $x_1 = x_2$.)
\end{remark}

\section*{Acknowledgement}
The authors thank Matthias Mnich for valuable comments. Yury Makarychev was supported by NSF CAREER award CCF-1150062 and NSF award IIS-1302662.

\appendix

\section{Rankwise independent permutations}\label{sec:rankwise}
In this section, we prove the following lemma.

\begin{lemma}\label{lem:rankwise}
If $\tilde{\alpha}$ is a random $4k$ rankwise independent permutation and $|V'|\geq \kappa_k t^2$, then for some $\alpha^*$ in the
support of $\tilde{\alpha}$, $\val_{\calI}(\alpha^*)\geq AVG + t$.
\end{lemma}
\begin{proof}
Let $\alpha$ be a permutation uniformly distributed among all $n!$ permutations. The random variables $\val_{\calI}(\alpha)$ and $\Phi$
are identically distributed. Hence,
$$\Var[\val_{\calI}(\alpha)] = \Var[\Phi],\;\;\; \text{and}\;\;\;
\|\val_{\calI}(\alpha)- AVG\|_4 = \|\Phi - AVG\|_4.$$ Observe, that
$$\Var[\val_{\calI}(\alpha^*)] = \Var[\val_{\calI}(\alpha)],\;\;\; \text{and}\;\;\;
\|\val_{\calI}(\alpha)- AVG\|_4 = \|\val_{\calI}(\alpha^*)- AVG\|_4,$$
since for every four predicates $\pi_1,\pi_2,\pi_3,\pi_4\in \Pi$, we have
\begin{eqnarray*}
\E{\pi_1(\alpha^*)\pi_2(\alpha^*)} &=& \E{\pi_1(\alpha)\pi_2(\alpha)};\\
\E{\pi_1(\alpha^*)\pi_2(\alpha^*)\pi_3(\alpha^*)\pi_4(\alpha^*)} &=& \E{\pi_1(\alpha)\pi_2(\alpha)\pi_3(\alpha)\pi_4(\alpha)}.
\end{eqnarray*}
Hence, as in Theorem~\ref{thm:main}, $\Var[\val_{\calI}(\alpha^*)]\geq |V'|\beta_k/k$ and
$\|\val_{\calI}(\alpha^*) - AVG\|_4^4 \leq 81^k C_k \|\val_{\calI}(\alpha^*) - AVG\|_2^4$.
Consequently, by Theorem~\ref{thm:AGKSY}, $\Pr (\val_{\calI}(\alpha^*)\geq AVG + t) > 0$.
This concludes the proof.
\end{proof}


\begin{thebibliography}{WW}
\bibitem{AGKSY} N. Alon, G. Gutin, E. J. Kim, S. Szeider, and A. Yeo. Solving Max $r$-SAT Above a Tight Lower Bound. Algorithmica 61 (3), pp.~638--655 (2011).	

\bibitem{AL}
N. Alon and S. Lovett. Almost $k$-wise vs. $k$-wise independent permutations, and uniformity for general group actions.
Approximation, Randomization, and Combinatorial Optimization. Algorithms and Techniques. Springer Berlin Heidelberg, pp.~350--361, 2012.

\bibitem{AM}
P. Austrin and E. Mossel. Approximation resistant predicates from pairwise independence. Computational Complexity, 18(2), pp.~249--271, 2009.


\bibitem{Bonami}
A. Bonami. \'{E}tude des coefficients Fourier des fonctions de $L^p(G)$. Annales de
l' Institut Fourier, 20(2):335--402, 1970.

\bibitem{Chan}
S. Chan. Approximation resistance from pairwise independent subgroups.
STOC 2014, pp.~447--456.

\bibitem{CMM}
M. Charikar, K. Makarychev, and Y. Makarychev.
On the Advantage over Random for Maximum Acyclic Subgraph.
FOCS 2007, pp.~625--633.



\bibitem{CFGJRTY} R. Crowston, M. Fellows, G. Gutin, M. Jones, F. Rosamond, S. Thomass{\'e} and A. Yeo.
Simultaneously Satisfying Linear Equations Over $\mathbb{F}_2$: MaxLin2 and Max-$r$-Lin2
Parameterized Above Average. In FSTTCS 2011, LIPICS Vol. 13, 229--240.

\bibitem{CGJRS}
R. Crowston, G. Gutin, M. Jones, V. Raman, and S. Saurabh. Parameterized complexity of MaxSat above average. In the Proceedings of LATIN 2012, pp. 184--194.

\bibitem{CS}
B.~Chor and M.~Sudan. A geometric approach to betweenness. SIAM Journal on Discrete Mathematics, vol. 11, no. 4 (1998): pp.~511--523.

\bibitem{DS96}
P. Diaconis and L. Saloff-Coste. Logarithmic Sobolev inequalities for finite Markov chains. Ann. Appl. Probab.~6
(1996), no. 3, 695--750.

\bibitem{ES} B.~Efron and C.~Stein. The jackknife estimate of variance. Annals of
Statistics, 9(3):586--596, 1981.

\bibitem{ODonnell} R. O'Donnell. Analysis of Boolean Functions. Cambridge University Press. 2014. ISBN 9781107038325.

\bibitem{GHMRC}
V. Guruswami, J. H\aa{}stad, R. Manokaran, P. Raghavendra, and M. Charikar. Beating the random ordering is hard: Every ordering CSP is approximation resistant. SIAM Journal on Computing 40, no. 3 (2011): 878--914.


\bibitem{GR}
V.~Guruswami and P.~Raghavendra. Constraint Satisfaction over a Non-Boolean Domain:
Approximation Algorithms and Unique-Games Hardness. APPROX 2008, pp.~77-90.

\bibitem{GZ}
V. Guruswami and Y. Zhou.
Approximating Bounded Occurrence Ordering CSPs. APPROX-RANDOM 2012, pp.~158--169.

\bibitem{GIMY}
G. Gutin, L. van Iersel, M. Mnich, and A. Yeo. All Ternary Permutation Constraint Satisfaction Problems Parameterized Above Average Have Kernels with Quadratic Number of Variables. J. Comput. Syst. Sci. 78 (2012), 151--163.

\bibitem{GKMY}
G. Gutin, E. J. Kim, M. Mnich, and A. Yeo. Betweenness parameterized above tight lower bound. J. Comput. Syst. Sci., 76: 872--878, 2010.

\bibitem{GKSY}
G. Gutin, E.J. Kim, S. Szeider, and A. Yeo. A Probabilistic Approach to Problems Parameterized
Above or Below Tight Bounds. J. Comput. Syst. Sci. 77 (2011), 422--429.


\bibitem{GY} G.~Gutin and A.~Yeo. Constraint satisfaction problems parameterized above or below tight bounds: a survey. In The Multivariate Algorithmic Revolution and Beyond,
pp.~257--286, 2012.

\bibitem{Hastad}
J. H{\aa}stad. Some optimal inapproximability results. STOC 1997.


\bibitem{ITT}
T.~Itoh, Y.~Takei, and J.~Tarui.
On Permutations with Limited Independence.
SODA 2000, pp.~137--146.

\bibitem{Karp72}
R.~Karp.
Reducibility among combinatorial problems. In Complexity of Computer Computations.
New York: Plenum, 1972, pp. 85--103.

\bibitem{KW} E.~J.~Kim and R.~Williams. Improved parameterized algorithms for above average constraint satisfaction. In Parameterized and Exact Computation, pp. 118--131, 2012.

\bibitem{MR} V. Mahajan and V. Raman. Parameterizing above Guaranteed Values: MaxSat and MaxCut. Journal of Algorithms, Vol. 31, Issue 2, May 1999, pp. 335--354.

\bibitem{Mak} K.~Makarychev. Local Search is Better than Random Assignment for Bounded Occurrence Ordering $k$-CSPs. STACS 2013, pp.~139--147.

\bibitem{M-betw} Y.~Makarychev. Simple linear time approximation algorithm for betweenness. Operations Research Letters, vol.~40, no. 6 (2012), pp.~450---452.

\bibitem{MOS}
E.~Mossel, K.~Oleszkiewicz and A.~Sen.
On Reverse Hypercontractivity
Geometric and Functional Analysis, vol 23(3), pp.~1062--1097, 2013.

\bibitem{Opatrny79}
J.~Opatrny. Total ordering problem. SIAM Journal on Computing, 8(1):111--114, Feb. 1979.

\bibitem{RS} I. Razgon and B. O'Sullivan. Almost 2-SAT is fixed-parameter tractable. J. Comput. Syst.
Sci. 75(8):435--450, 2009.

\bibitem{Seymour} P.~D.~Seymour. Packing directed circuits fractionally. Combinatorica, vol.~15, no. 2 (1995), pp.~281--288.

\bibitem{Talagrand94}
M. Talagrand. On Russo’s approximate 0-1 law , Annals of Probability 22 (1994), 1576--1587.

\bibitem{Wolf07}
P.~Wolff. Hypercontractivity of simple random variables, Studia Mathematica
180 (2007), pp.~219--326.
\end{thebibliography}
\end{document}